\documentclass[acmsmall,screen]{acmart}

\setcopyright{rightsretained}
\acmPrice{}
\acmDOI{10.1145/3498671}
\acmYear{2022}
\copyrightyear{2022}
\acmSubmissionID{popl22main-p48-p}
\acmJournal{PACMPL}
\acmVolume{6}
\acmNumber{POPL}
\acmArticle{10}
\acmMonth{1}

\usepackage{booktabs}   
\usepackage{subcaption} 

\usepackage{xcolor}
\usepackage[inline]{enumitem}
\usepackage{pifont}
\usepackage[ruled, noline, boxed, commentsnumbered, longend]{algorithm2e}
\usepackage[skins]{tcolorbox}
\usepackage{tabu}
\usepackage{tikz}
\usetikzlibrary{shapes,shapes.misc,arrows.meta,automata,positioning,arrows,decorations.markings}
\usepackage[capitalize,nameinlink,noabbrev]{cleveref}
\usepackage{listings}
\usepackage{multirow}
\usepackage{amsmath}
\usepackage{microtype}

\newif\ifappendix
\appendixtrue

\ifappendix
\newcommand{\appref}[1]{the appendix}
\else
\newcommand{\appref}[1]{~\cite{#1}}
\fi

\crefformat{section}{\S#2#1#3}
\crefformat{subsection}{\S#2#1#3}
\crefformat{subsubsection}{\S#2#1#3}

\newcommand{\cons}{\mathit{cons}}
\newcommand{\cdr}{\mathit{tail}}
\newcommand{\car}{\mathit{head}}
\newcommand{\nil}{\mathit{nil}}

\newcommand{\many}[1]{\bar{#1}}
\newcommand{\arity}[1]{\mathsf{arity}(#1)}

\newcommand{\alphabet}{\Sigma}
\newcommand{\graph}{\mathsf{graph}}

\newcommand{\Lat}{Lat}
\newcommand{\Def}{\mathsf{Def}}
\newcommand{\Block}{\mathsf{Block}}

\newcommand{\MyAnd}{\mathsf{and}}
\newcommand{\Let}{\mathsf{let}}
\newcommand{\In}{\mathsf{in}}

\newcommand{\IfThenElse}[3]{\iteterm(#1,#2,#3)}
\newcommand{\ite}{if-then-else}
\newcommand{\iteterm}{\mathsf{ite}}

\newcommand{\EXP}{\mathsf{EXPTIME}}

\newcommand{\FO}{\mathsf{FO}}
\newcommand{\FOk}{\FO(k)}

\newcommand{\FOLFP}{\FO\text{-}\LFP}
\newcommand{\LFP}{\mathsf{LFP}}
\newcommand{\FOTERM}{\FO\text{-}\LFP\text{-}\mathsf{FUN}}

\newcommand{\FOTERMk}{\FOTERM(k,k')}
\newcommand{\SOk}{\SO(k,k')}
\newcommand{\Datalog}{\mathsf{Datalog}}
\newcommand{\Datalogk}{\mathsf{Datalog}(k,k')}

\newcommand{\FOLFPk}{\FOLFP(k,k')}
\newcommand{\lfp}{\mathsf{lfp}}

\newcommand{\sepsynth}[1]{#1{-}\emph{separator realizability and synthesis}}
\newcommand{\querysepsynth}[1]{#1{-}\emph{query realizability and synthesis}}
\newcommand{\termsynth}[1]{#1{-}\emph{term synthesis}}

\newcommand{\sepreal}[1]{#1{-}\emph{separator realizability}}
\newcommand{\querysepreal}[1]{#1{-}\emph{query realizability}}

\newcommand{\MSO}{\mathsf{MSO}}
\newcommand{\SO}{\mathsf{SO}}

\newcommand{\grammar}{G}
\newcommand{\dom}{\mathsf{dom}}
\newcommand{\la}{\langle}
\newcommand{\ra}{\rangle}
\newcommand{\Nat}{\mathbb{N}}
\newcommand{\Bb}{\mathcal{B}}
\newcommand{\Oo}{\mathcal{O}}
\newcommand{\Pos}{\mathit{Pos}}
\newcommand{\Neg}{\mathit{Neg}}

\newlength\mylen
\newcommand\myinput[1]{%
  \settowidth\mylen{\KwIn{}}%
  \setlength\hangindent{\mylen}%
  \hspace*{\mylen}#1\\}
\newcommand\myoutput[1]{%
  \settowidth\mylen{\KwOut{}}%
  \setlength\hangindent{\mylen}%
  \hspace*{\mylen}#1\\}

\newcommand{\andsymb}{\wedge}
\newcommand{\orsymb}{\vee}
\newcommand{\negsymb}{\neg}
\newcommand{\allsymb}[1]{\forall{#1}}
\newcommand{\existsymb}[1]{\exists{#1}}
\newcommand{\update}[3]{#1[#2\mapsto#3]}
\newcommand{\tru}{\mathsf{True}}
\newcommand{\fals}{\mathsf{False}}
\newcommand{\aut}{\mathcal{A}}
\newcommand{\treepos}{\mathit{Nodes}}
\newcommand{\tree}{\rho}

\definecolor{zgray}{RGB}{63,63,63}
\definecolor{zyellow}{RGB}{240,223,175}
\definecolor{zgreen}{RGB}{143,175,159}
\definecolor{zcream}{RGB}{220,220,204}
\definecolor{zpink}{RGB}{220,163,163}
\definecolor{zblue}{RGB}{140,208,211}
\definecolor{assassinblue}{RGB}{110,150,210}
\definecolor{assassinred}{RGB}{192,106,111}

\newcommand{\problabel}[1]{\label{prob:#1}}
\newcommand{\probref}[1]{Problem~\ref{prob:#1}}

\newcommand{\production}[1]{#1}
\newcommand{\produces}{\Coloneqq}
\newcommand{\conj}[2]{#1\,\wedge\,#2}
\newcommand{\ifthen}[2]{\mathsf{if}\,#1\,\mathsf{then}\,#2}
\newcommand{\iif}[1]{\mathsf{if}\,#1}
\newcommand{\then}[1]{\mathsf{then}\,#1}
\newcommand{\eelse}[1]{\mathsf{else}\,#1}
\newcommand{\elifthen}[2]{\mathsf{elif}\,#1\,\mathsf{then}\,#2}

\theoremstyle{plain}
\newtheorem{theorem}{Theorem}
\theoremstyle{definition}
\newtheorem*{claim*}{Claim}

\newtheorem{lemma}[theorem]{Lemma}

\newcommand{\tableref}[1]{Table~\ref{tab:#1}}
\newcommand{\figlabel}[1]{\label{fig:#1}}
\newcommand{\figref}[1]{Figure~\ref{fig:#1}}
\newcommand{\seclabel}[1]{\label{sec:#1}}
\newcommand{\secref}[1]{Section~\ref{sec:#1}}
\newcommand{\csecref}[1]{\cref{sec:#1}}

\newcommand{\thmlabel}[1]{\label{thm:#1}}
\newcommand{\thmref}[1]{Theorem~\ref{thm:#1}}

\newcommand{\lemlabel}[1]{\label{lem:#1}}
\newcommand{\lemref}[1]{Lemma~\ref{lem:#1}}

\newcommand{\equlabel}[1]{\label{eq:#1}}
\newcommand{\equref}[1]{Equation~(\ref{eq:#1})}
\newcommand{\applabel}[1]{\label{app:#1}}

\newcommand{\Ll}{\mathcal{L}}

\renewcommand{\emptyset}{\varnothing}

\newcommand{\xdownarrow}[1]{%
  {\left\downarrow\vbox to #1{}\right.\kern-\nulldelimiterspace}
}

\newcommand{\dblqt}[1]{\text{``}#1\text{''}}

\begin{document}

\title[Learning Formulas in Finite Variable Logics]{Learning Formulas in Finite Variable Logics}

\author{Paul Krogmeier}
\email{paulmk2@illinois.edu} 
\affiliation{ \department{Department of Computer
    Science} 
  \institution{University of Illinois, Urbana-Champaign} 
  \country{USA} 
}

\author{P. Madhusudan}
\email{madhu@illinois.edu} 
\affiliation{ \department{Department of Computer
    Science} 
  \institution{University of Illinois,
    Urbana-Champaign} 
  \country{USA} 
}


\begin{abstract}
  We consider grammar-restricted exact learning of formulas and terms
  in finite variable logics. We propose a novel and versatile
  automata-theoretic technique for solving such problems. We first
  show results for learning formulas that classify a set of
  positively- and negatively-labeled structures. We give algorithms
  for realizability and synthesis of such formulas along with upper
  and lower bounds. We also establish positive results using our
  technique for other logics and variants of the learning problem,
  including first-order logic with least fixed point definitions,
  higher-order logics, and synthesis of queries and terms with
  recursively-defined functions.
\end{abstract}



\maketitle
\renewcommand{\shortauthors}{P. Krogmeier, P. Madhusudan}

\section{Introduction}\seclabel{intro}
Learning symbolically representable concepts from data is an important
emerging area of research. Symbolic expressions, such as logical
formulas or programs, can be easily analyzed and interpreted, which
aids downstream applications (e.g., analyzing a large system that has
a classifier as a component) and makes them easier to communicate to
both humans and computers.

In this paper, we embark on a foundational study of \emph{exact}
learning of logical formulas. For a logic $\Ll$, we study the
\emph{separability} problem: given a set of positively- and
negatively-labeled finite structures, we want to learn a sentence
$\varphi$ in $\Ll$ that is true on the positive structures and false
on the negative structures.  Separability consists of two related
problems. First, the \emph{realizability problem}: a decision problem
that asks whether such a sentence exists, and second, the
\emph{synthesis problem}, which asks to construct a sentence if one
exists. For logics that contain \emph{infinitely many} semantically
inequivalent formulas, including most of the logics considered in this
paper, the realizability problem itself is not trivial.

In a learning context, one is often interested in how well a learned
artifact generalizes to unseen inputs. In practice, most learning
algorithms typically attempt only to minimize loss in accuracy on a
set of training samples~\cite{mitchell97}. Exact learning asks for a
perfect classifier with respect to the training samples. Two common
strategies to mitigate \emph{overfitting} are (1) to only consider
classifiers from a restricted hypothesis class ${\mathcal H}$ and (2)
to prefer \emph{simple} concepts over complex ones. The problems we
study here reflect (1) by considering exact learning with respect to
\emph{grammars}. Problem instances are equipped with a grammar
$\grammar$ that defines a subset $L(\grammar)$ of logical expressions
in $\Ll$ to which classifiers must belong. The problems reflect (2) by
requiring a synthesizer to construct small (perhaps the smallest)
formulas that separate sample structures.

We describe a very general technique for solving the realizability and
synthesis problems for \emph{several} logics with finitely many
variables. In particular, a main contribution of our work is to solve
realizability and synthesis for $\FOk$, a version of first-order logic
with $k$ variables. This logic allows for an arbitrary number and
nesting depth of quantifiers. Although the number of variables is
bounded, it is possible to \emph{reuse} variables. For instance,
consider finite graphs. Given two constants $s$ and $t$ and any
$n \in \mathbb{N}$, the property that \emph{$t$ is reachable from $s$
  using at most $n$ edges} is expressible using just two variables,
and thus $\FOk$ with $k=2$ contains an infinite set of inequivalent
formulas.

We prove that for every $k \in \mathbb{N}$, the realizability problem
for $\FOk$ over a grammar is decidable. That is, given a (tree)
grammar $\grammar$ defining a subclass of $\FOk$ and sets of
positively- and negatively-labeled structures $\mathit{Pos}$ and
$\mathit{Neg}$, it is decidable to check whether there is a sentence
$\varphi\in L(\grammar)$ that is true on all structures in
$\mathit{Pos}$ and false on all structures in $\mathit{Neg}$. We give
an algorithm to synthesize such a sentence if one exists.
Notice that since structures are finite, the sentence can also be converted to a
program operating over structures that realizes the classifier
effectively.

\smallskip

\noindent \emph{Automata over Parse Trees for Realizability and Synthesis.}  Our
primary technique for solving exact learning problems of this kind is
based on automata over finite trees. Intuitively, given finite
(disjoint) sets of positive and negative structures, we need to search
through an infinite set of formulas that adhere to the grammar
$\grammar$ in order to find a separating formula. We use tree automata
working over \emph{formula parse trees} and show that the set of all
separating formulas for $\mathit{Pos}$ and $\mathit{Neg}$ that adhere
to $\grammar$ forms a regular class of trees. Building the tree
automaton and checking emptiness gives us a decision procedure for
realizability. Algorithms for tree automaton emptiness are used to
solve the synthesis problem. Furthermore, the algorithms can be
adapted to find the \emph{smallest} trees that are accepted by the
tree automaton, hence giving us small formulas as separators.

The key idea is to show that, given a single structure $A$, the set of
parse trees for \emph{all} sentences adhering to the grammar
$\grammar$ that are true (or false) in $A$ is a regular language. Given
$A$, we can define an automaton that interprets an input formula
(rather, its parse tree) on $A$ and checks whether $A$ satisfies the
formula. This evaluation follows the usual semantics of the logic in
question, which is typically defined \emph{recursively}, and hence can
be evaluated bottom-up in the structure of the formula. In general,
this requires simulating the semantics of formulas for \emph{each}
assignment to free variables over the structure $A$. The bounded
variable restriction is therefore crucial to ensure that the automaton
only needs an amount of state that depends on the size of the
structure but is \emph{independent} of the size of the formula.

We can then construct an automaton that captures precisely the set of
separating formulas for the given labeled structures. We do this by
constructing automata for (a) the set of all formulas that are true on
positive structures, (b) the set of all formulas that are false on
negative structures, and (c) the intersection of the automata from (a)
and (b). We further intersect (c) with an automaton that accepts
formulas allowed by the grammar $\grammar$. Checking emptiness of this
final automaton solves realizability, and we can construct an accepted
formula if nonempty.

\smallskip

\noindent
\emph{Query and Term Synthesis.} We study two related learning
problems in addition to the separability problem. We study
\emph{query}\, synthesis, where we are given a grammar $\grammar$ and a
finite set of structures, each accompanied by an \emph{answer set}\, of
$r$-tuples from the domain of the structure. The query synthesis
problem is to find a query, namely, a first-order formula
$\varphi\in L(\grammar)$ with $r$ free variables, such that the sets
of tuples that satisfy $\varphi$ in each structure are precisely the
given answer sets. We also study the problem of \emph{term synthesis},
which is closer in spirit to program synthesis from input-output
examples (using logic as a programming language). In this problem we
are given a set of input structures and a grammar $\grammar$, and each
structure interprets a set of constants, e.g.,
$\mathit{in_1},\ldots, \mathit{in_d}$ and $\mathit{out}$, as a
particular input-output example. The term synthesis problem is to find
a closed first-order \emph{term} $t$ such that $t$ evaluates to
$\mathit{out}$ in each structure. The grammar can ensure that
$\mathit{out}$ is not used in $t$. Again, we note that such
first-order queries and terms can also be realized effectively as
programs that operate over structures: a program for a query
instantiates the free variables with all $r$-tuples, evaluates the
formula, and returns those that satisfy it. Similarly, a term can be
converted to a program that recursively evaluates its subterms and
returns an element of the structure. We give adaptations of the
automata-theoretic technique to solve both the query and term
synthesis problems for the logic $\FOk$.

\smallskip

\noindent
\emph{Learning Algorithms for First-Order Logic with Least Fixed
  Points.} A second contribution of this work is showing that a
bounded variable version of first-order logic with least fixed points
also has decidable separator realizability and synthesis. Least fixed
point definitions add the power of recursion to first-order logic,
resulting in a more expressive logic that, for instance, can describe
transitive closure of relations and the semantics of languages like
$\Datalog$.  Furthermore, over finite ordered structures, first-order
logic with least fixed points captures the class $\mathbb{P}$ of all
functions computable in polynomial time. Consequently, learned
formulas in this logic can be realized as polynomial-time programs.

In this case we use \emph{two-way alternating tree automata} to
succinctly encode the semantics of expressions with least fixed point
definitions over a given structure. Intuitively, we need to use
recursion to evaluate a recursively-defined relation, and in each
recursive call we need to \emph{read} the definition of the relation
once more. This capability is elegantly provided by two-way automata,
and alternation gives a way to compositionally send copies of the
automaton to check various sub-formulas effectively. Two-way
alternating automata can be converted to one-way nondeterministic
automata (with an exponential increase in states), and emptiness
checking for the resulting automata gives the algorithm we seek for
separator realizability and synthesis. We also solve the term
synthesis problem for a logic with least fixed points, a problem which
resembles functional program synthesis.

\smallskip
\noindent \emph{Further Results.} A remarkable aspect of the automata-theoretic
approach is that it provides algorithms for synthesis in many
settings.
The constructions smoothly extend to virtually any logic
where the bounded variable restriction yields a formula evaluation
strategy with a memory requirement that is independent of the size of
the formula. In~\cref{sec:further-results} we discuss decidable
realizability and synthesis results for other settings where related
problems have been studied, including languages with mutual recursion,
e.g., $\Datalog$~\cite{aws-synth-datalog,evans-greffen-noisy} and
inductive logic programming~\cite{ilp,ilp-turning30,mil-muggleton}. We
also discuss how the technique extends to higher-order logics over
finite models.

\smallskip
\noindent
\emph{Complexity.} For each logic we consider, we present the upper
bound that the automaton construction yields in terms of various
parameters: the number of variables $k$, the maximum size of each
structure, the number of structures, and the size of the grammar
$\grammar$. A sample of upper bounds is given in~\tableref{results}
in~\csecref{problems}. We also prove lower bounds for the logic
$\FOk$, arguing that the complexity of the upper bounds on certain
parameters is indeed tight. In particular, we show that for fixed $k$,
separator realizability is $\EXP$-hard. This matches the upper bound,
and also proves matching lower bounds for more expressive logics.

In summary, our work provides an extremely general tree
automata-theoretic technique that yields effective solutions for the
problems of learning separators/queries/terms for several
finite-variable logics. Our contributions here are theoretical. We
establish decidability for several exact learning problems over
different logics, and we give algorithms based on tree automata and
some matching lower bounds. We believe and hope that this work will
inform the design of practical algorithms for these problems. In
particular, our automata constructions and the ``bottom up'' fixed
point procedure for checking automaton emptiness gives a design
framework for such algorithms. Due to the relatively high worst-case
complexity of synthesis, practical algorithms will need to adapt to
application domains and cater to restricted logics and languages that
admit more efficient synthesis, potentially using heuristics,
space-efficient data structures, and fast search (e.g., BDDs, SAT
solvers, etc.)  (see~\cite{gr1,Wang2017,WangWangDilligOOPSLA18} for
examples).


\ifappendix
Details for proofs and automata constructions can be found
in the appendix.
\else
Details for proofs and automata constructions can be found
in the extended version~\cite{fullversion}.
\fi

\section{Examples}
\seclabel{examples}

We begin with some examples to illustrate instances of the exact
learning problems considered here. The first two examples explore the
separability problem for first-order logic and the subtlety around
reusing variables in bounded variable logics. The third example
illustrates learning formulas with least fixed point definitions, and
the fourth example illustrates term synthesis.

\subsection{Example 1: Learning Formulas, Significance of Grammar, and
  Unrealizability}
\label{sec:example-1}
Consider the problem of finding a separating first-order sentence for
the structures depicted in~\figref{example1} using a vocabulary that
includes the binary edge relation $E$ and constants $s$ and $t$.

\vspace{-0.01in}

\tikzstyle{every node}=[circle, text=zgray, draw, inner sep=1.3pt, minimum width=1pt]
\tikzset{every loop/.style={in=270,out=180,looseness=5}}
\begin{figure}[H]
  \centering
  \begin{tabular}[t]{p{0.19\textwidth}>{\centering}p{0.19\textwidth}>{\centering}p{0.01\textwidth}>{\centering}p{0.19\textwidth}>{\centering}p{0.19\textwidth}}
  \scalebox{0.9}{
    \centering
    \begin{tikzpicture}[thick,scale=0.8]
      \node[shape=circle,draw=zgray,fill=zgray] (0) at (0,0) {} ;
      \node[shape=circle,draw=zgray,fill=zgray] (1) at (1,0) {} ;
      \node[shape=circle,draw=zgray,fill=zgray] (2) at (2,0) {} ;
      \node[shape=circle,draw=zgray,fill=zgray] (3) at (0,1) {} ;
      \node[shape=circle,draw=zgray,fill=zgray] (4) at (1,1) {} ;
      \node[shape=circle,draw=zgray,fill=zgray] (5) at (2,1) {} ;
      \node[shape=circle,draw=zgray,fill=zgray,label=below:$s$] (6) at (1,-1) {} ;
      \node[shape=circle,draw=zgray,fill=zgray,label=above:$t$] (7) at (1,2) {} ;
      \node[scale=1.3,draw=none] at (-0.5,0.5) {\texttt{+}};

      \draw[fill=zgray] [-] (0) edge[zgray] (3) ;
      \draw[fill=zgray] [-] (1) edge[zgray] (3) ;
      \draw[fill=zgray] [-] (6) edge[zgray] (0) ;
      \draw[fill=zgray] [-] (6) edge[zgray] (1) ;
      \draw[fill=zgray] [-] (3) edge[zgray] (7) ;
      \draw[fill=zgray] [-] (4) edge[zgray] (7) ;
      \draw[fill=zgray] [-] (5) edge[zgray] (7) ;
    \end{tikzpicture}
    \vspace{0.1in}
  } &
    \scalebox{0.9}{
    \centering
    \begin{tikzpicture}[thick,scale=0.8]
      \node[shape=circle,draw=zgray,fill=zgray] (0) at (0,0) {} ;
      \node[shape=circle,draw=zgray,fill=zgray] (1) at (1,0) {} ;
      \node[shape=circle,draw=zgray,fill=zgray] (2) at (2,0) {} ;
      \node[shape=circle,draw=zgray,fill=zgray] (3) at (0,1) {} ;
      \node[shape=circle,draw=zgray,fill=zgray] (4) at (1,1) {} ;
      \node[shape=circle,draw=zgray,fill=zgray] (5) at (2,1) {} ;
      \node[shape=circle,draw=zgray,fill=zgray,label=below:$s$] (6) at (1,-1) {} ;
      \node[shape=circle,draw=zgray,fill=zgray] (8) at (0,-1) {} ;
      \node[shape=circle,draw=zgray,fill=zgray] (9) at (2,-1) {} ;
      \node[shape=circle,draw=zgray,fill=zgray,label=above:$t$] (7) at (1,2) {} ;
      \node[scale=1.3,draw=none] at (-0.5,0.5) {\texttt{+}};

      \draw[fill=zgray] [-] (0) edge[zgray] (3) ;
      \draw[fill=zgray] [-] (1) edge[zgray] (4) ;
      \draw[fill=zgray] [-] (2) edge[zgray] (5) ;
      \draw[fill=zgray] [-] (3) edge[zgray] (7) ;
      \draw[fill=zgray] [-] (4) edge[zgray] (7) ;
      \draw[fill=zgray] [-] (5) edge[zgray] (7) ;
      \draw[fill=zgray] [-] (8) edge[zgray] (4) ;
      \draw[fill=zgray] [-] (9) edge[zgray] (4) ;
      \draw[fill=zgray] [-] (8) edge[zgray] (0) ;
      \draw[fill=zgray] [-] (9) edge[zgray] (2) ;
    \end{tikzpicture}
    \vspace{0.1in}
      } & &
    \scalebox{0.9}{
    \centering
    \begin{tikzpicture}[thick,scale=0.8]
      \node[shape=circle,draw=zgray,fill=zgray] (0) at (0,0) {} ;
      \node[shape=circle,draw=zgray,fill=zgray] (1) at (1,0) {} ;
      \node[shape=circle,draw=zgray,fill=zgray] (2) at (2,0) {} ;
      \node[shape=circle,draw=zgray,fill=zgray] (3) at (0,1) {} ;
      \node[shape=circle,draw=zgray,fill=zgray] (4) at (1,1) {} ;
      \node[shape=circle,draw=zgray,fill=zgray] (5) at (2,1) {} ;
      \node[shape=circle,draw=zgray,fill=zgray,label=below:$s$] (6) at (1,-1) {} ;
      \node[shape=circle,draw=zgray,fill=zgray,label=above:$t$] (7) at (1,2) {} ;
      \node[scale=1.3,draw=none] at (-0.5,0.5) {\texttt{-}};

      \draw[fill=zgray] [-] (0) edge[zgray] (3) ;
      \draw[fill=zgray] [-] (1) edge[zgray] (3) ;
      \draw[fill=zgray] [-] (6) edge[zgray] (0) ;
      \draw[fill=zgray] [-] (6) edge[zgray] (1) ;
      \draw[fill=zgray] [-] (6) edge[zgray] (2) ;
      \draw[fill=zgray] [-] (3) edge[zgray] (7) ;
      \draw[fill=zgray] [-] (4) edge[zgray] (7) ;
      \draw[fill=zgray] [-] (5) edge[zgray] (7) ;
    \end{tikzpicture}
    \vspace{0.1in}
      }
   &
  \scalebox{0.9}{
    \centering
    \begin{tikzpicture}[thick,scale=0.8]
      \node[shape=circle,draw=zgray,fill=zgray] (0) at (0,0) {} ;
      \node[shape=circle,draw=zgray,fill=zgray] (1) at (1,0) {} ;
      \node[shape=circle,draw=zgray,fill=zgray] (2) at (2,0) {} ;
      \node[shape=circle,draw=zgray,fill=zgray] (3) at (0,1) {} ;
      \node[shape=circle,draw=zgray,fill=zgray] (4) at (1,1) {} ;
      \node[shape=circle,draw=zgray,fill=zgray] (5) at (2,1) {} ;
      \node[shape=circle,draw=zgray,fill=zgray,label=below:$s$] (6) at (1,-1) {} ;
      \node[shape=circle,draw=zgray,fill=zgray,label=above:$t$] (7) at (1,2) {} ;
      \node[scale=1.3,draw=none] at (-0.5,0.5) {\texttt{-}};

      \draw[fill=zgray] [-] (0) edge[zgray] (3) ;
      \draw[fill=zgray] [-] (1) edge[zgray] (3) ;
      \draw[fill=zgray] [-] (2) edge[zgray] (5) ;
      \draw[fill=zgray] [-] (6) edge[zgray] (0) ;
      \draw[fill=zgray] [-] (6) edge[zgray] (1) ;
      \draw[fill=zgray] [-] (6) edge[zgray] (2) ;
      \draw[fill=zgray] [-] (3) edge[zgray] (4) ;
      \draw[fill=zgray] [-] (4) edge[zgray] (7) ;
      \draw[fill=zgray] [-] (5) edge[zgray] (7) ;
    \end{tikzpicture}
    \vspace{0.1in}
      }
  \end{tabular}
  \vspace{-0.1in}
  \caption{Find a sentence in first-order logic that is true for \texttt{+}
    structures and false for \texttt{-} structures.
    }
  \figlabel{example1}
\end{figure}
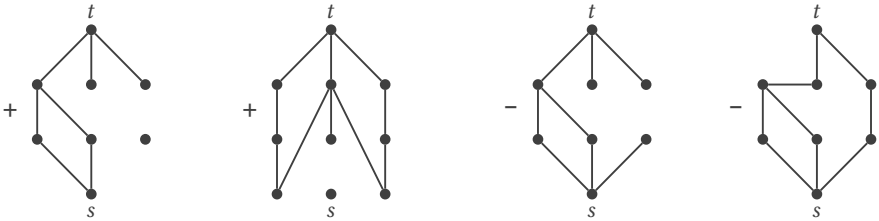

\vspace{-0.1in}
One possible solution asserts that \emph{every node that is adjacent
  to $s$ is adjacent to a node that is adjacent to $t$}. In
first-order logic:
\begin{align*}
  \forall x.\, \left(E(s,x)\rightarrow\exists y.\, \left(E(x,y)\wedge E(y,t)\right)\right)
\end{align*}

If the grammar $\grammar$ allowed, say, all first-order logic formulas
with two variables, the formula above would indeed be a solution. If
instead the grammar allowed only conjunction as a Boolean connective,
then the formula above is not a separator. If the grammar allowed only
one variable, or allowed two variables but disallowed universal
quantification, then there is no separator.

In fact, for a grammar that only allows conjunction as Boolean
connective, there is \emph{no} separator for the structures
above. \emph{Proof gist.  }For a contradiction: suppose there is a
separator $\varphi$ that does not use negation or disjunction. Then
$\varphi$ has a positive matrix (inner formula has no negations), and
its standard conversion to an equivalent formula in prenex form,
$\mathit{prenex}(\varphi)$, still has a positive matrix and will be a
separator (though it may have many more variables).  Since
$\mathit{prenex}(\varphi)$ has a positive matrix, any graph that
satisfies it will continue to satisfy it if we add any number of
edges. Since the leftmost negative structure adds a single edge to the
leftmost positive structure, we have a contradiction. Thus there is no
separator.

Note that while algorithms can search for separators in $\grammar$,
the problem of declaring that there is \emph{no} separator is a
nontrivial problem, especially for arbitrary grammars. The algorithms
we seek will terminate and declare unrealizability when separators
do not exist (as in the above example).

\subsection{Example 2: Reuse of Variables and Infinite Semantic Concept Space}
\label{sec:example-2}
Consider the problem of finding a separator that uses only \emph{three
  variables} for the labeled structures in~\figref{example2}. One
possible solution is $\bigvee_{i=1}^7 \mathit{path}_i(s,t)$, where
$\mathit{path}_i(x,y)$ holds for elements $x,y$ if there is a
(directed) $E$-path of length $i$ from $x$ to $y$. Note that, by
\emph{reusing variables}, the formula $\mathit{path}_i(x,y)$ can be
defined using only $3$ variables for any $i\in\Nat$:
\begin{align*}
  \mathit{path}_1(x,y)\,\, &\leftrightarrow \,\, E(x,y) \\
  \mathit{path}_{i+1}(x,y)\,\, &\leftrightarrow \,\, \exists z.\, (E(x,z)\wedge
                       \exists x.\, (x=z\wedge \mathit{path}_{i}(x,y))) \quad i>0
\end{align*}
This example shows that even when we restrict the number of variables,
there is an infinite number of logically inequivalent sentences in
first-order logic (e.g., $\textit{path}_i$ for any $i>0$). But note
that this is not true if we bound the number of quantifiers or bound
the depth of quantifiers. This infinite semantic space of concepts is
what makes declaring unrealizability a nontrivial problem.  As we will
see, our technique works for finite variable logics in general,
despite the fact that they admit infinitely-many inequivalent
formulas.

\vspace{-0.1in}
\tikzstyle{every node}=[circle, text=zgray, draw, inner sep=1.3pt, minimum width=1pt]
\tikzset{every loop/.style={in=270,out=180,looseness=5}}
\begin{figure}[H]
  \centering
  \begin{tabular}[t]{p{0.18\textwidth}>{\centering}p{0.18\textwidth}>{\centering}p{0.01\textwidth}>{\centering}p{0.18\textwidth}>{\centering}p{0.18\textwidth}}
  \scalebox{0.9}{
    \centering
    \begin{tikzpicture}[thick,scale=0.8]
      \node[shape=circle,draw=zgray,fill=zgray] (0) at (0,1) {} ;
      \node[shape=circle,draw=zgray,fill=zgray,label=above:$t$] (1) at (0.27,1.669) {} ;
      \node[shape=circle,draw=zgray,fill=zgray,label=above:$s$] (3) at (1.745,1.658) {} ;
      \node[shape=circle,draw=zgray,fill=zgray] (4) at (2,1) {} ;
      \node[shape=circle,draw=zgray,fill=zgray] (5) at (1.78,0.39) {} ;
      \node[shape=circle,draw=zgray,fill=zgray] (6) at (1,0) {} ;
      \node[shape=circle,draw=zgray,fill=zgray] (7) at (0.265,0.337) {} ;
      \node[scale=1.3,draw=none] at (1,-0.5) {\texttt{+}};

      \draw[fill=zgray] [-to] (0) edge[zgray, bend left=10] (1) ;
      \draw[fill=zgray] [-to] (3) edge[zgray, bend left=10] (4) ;
      \draw[fill=zgray] [-to] (4) edge[zgray, bend left=10] (5) ;
      \draw[fill=zgray] [-to] (5) edge[zgray, bend left=10] (6) ;
      \draw[fill=zgray] [-to] (6) edge[zgray, bend left=10] (7) ;
      \draw[fill=zgray] [-to] (7) edge[zgray, bend left=10] (0) ;
    \end{tikzpicture}
    \vspace{0.1in}
  } &
    \scalebox{0.9}{
    \centering
    \begin{tikzpicture}[thick,scale=0.8]
      \node[shape=circle,draw=zgray,fill=zgray,label=left:$s$] (0) at (0,1) {} ;
      \node[shape=circle,draw=zgray,fill=zgray] (2) at (1,2) {} ;
      \node[shape=circle,draw=zgray,fill=zgray,label=right:$t$] (4) at (2,1) {} ;
      \node[shape=circle,draw=zgray,fill=zgray] (6) at (1,0) {} ;
      \node[scale=1.3,draw=none] at (1,-0.5) {\texttt{+}};

      \draw[fill=zgray] [-to] (0) edge[zgray, bend left=25] (2) ;
      \draw[fill=zgray] [-to] (2) edge[zgray, bend left=25] (4) ;
      \draw[fill=zgray] [-to] (4) edge[zgray, bend left=25] (6) ;
      \draw[fill=zgray] [-to] (6) edge[zgray, bend left=25] (0) ;
    \end{tikzpicture}
    \vspace{0.1in}
      } & &
    \scalebox{0.9}{
    \centering
    \begin{tikzpicture}[thick,scale=0.8]
      \node[shape=circle,draw=zgray,fill=zgray] (0) at (0,1) {} ;
      \node[shape=circle,draw=zgray,fill=zgray] (2) at (1,2) {} ;
      \node[shape=circle,draw=zgray,fill=zgray] (4) at (2,1) {} ;
      \node[shape=circle,draw=zgray,fill=zgray,label=left:$s$] (8) at (0,0) {} ;
      \node[shape=circle,draw=zgray,fill=zgray,label=right:$t$] (10) at (2,0) {} ;
      \node[scale=1.3,draw=none] at (1,-0.5) {\texttt{-}};

      \draw[fill=zgray] [-to] (0) edge[zgray, bend left=25] (2) ;
      \draw[fill=zgray] [-to] (8) edge[zgray] (0) ;
      \draw[fill=zgray] [-to] (4) edge[zgray] (10) ;
    \end{tikzpicture}
    \vspace{0.1in}
      }
   &
  \scalebox{0.9}{
    \centering
    \begin{tikzpicture}[thick,scale=0.8]
      \node[shape=circle,draw=zgray,fill=zgray] (0) at (0,1) {} ;
      \node[shape=circle,draw=zgray,fill=zgray] (1) at (0.27,1.669) {} ;
      \node[shape=circle,draw=zgray,fill=zgray,label=above:$t$] (2) at (1,2) {} ;
      \node[shape=circle,draw=zgray,fill=zgray,label=above:$s$] (8) at (1.745,2) {} ;
      \node[shape=circle,draw=zgray,fill=zgray] (3) at (2.5,1.658) {} ;
      \node[shape=circle,draw=zgray,fill=zgray] (4) at (2.745,1) {} ;
      \node[shape=circle,draw=zgray,fill=zgray] (5) at (2.525,0.39) {} ;
      \node[shape=circle,draw=zgray,fill=zgray] (6) at (1.745,0) {} ;
      \node[shape=circle,draw=zgray,fill=zgray] (9) at (1,0) {} ;
      \node[shape=circle,draw=zgray,fill=zgray] (7) at (0.265,0.337) {} ;
      \node[scale=1.3,draw=none] at (1.3725,-0.5) {\texttt{-}};

      \draw[fill=zgray] [-to] (0) edge[zgray, bend left=10] (1) ;
      \draw[fill=zgray] [-to] (1) edge[zgray, bend left=10] (2) ;
      \draw[fill=zgray] [-to] (8) edge[zgray, bend left=10] (3) ;
      \draw[fill=zgray] [-to] (3) edge[zgray, bend left=10] (4) ;
      \draw[fill=zgray] [-to] (4) edge[zgray, bend left=10] (5) ;
      \draw[fill=zgray] [-to] (5) edge[zgray, bend left=10] (6) ;
      \draw[fill=zgray] [-to] (6) edge[zgray, bend left=10] (9) ;
      \draw[fill=zgray] [-to] (9) edge[zgray, bend left=10] (7) ;
      \draw[fill=zgray] [-to] (7) edge[zgray, bend left=10] (0) ;
    \end{tikzpicture}
    \vspace{0.1in}
      }
  \end{tabular}
  \vspace{-0.1in}
  \caption{Find a sentence in first-order logic that is true on \texttt{+}
    structures and false on \texttt{-} structures. }
  \figlabel{example2}
\end{figure}
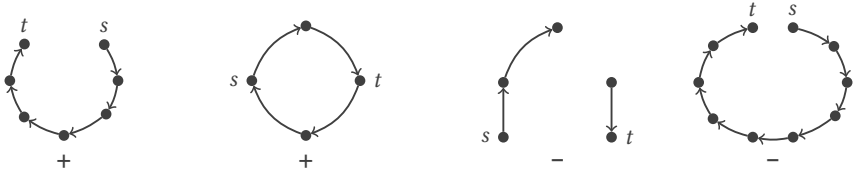

\subsection{Example 3: Least Fixed Points and Recursive
  Definitions}
\label{sec:example-3}
Notice that the separating sentence from \figref{example2} has a size
that depends on the sizes of the input structures, and it fails to
capture the notion of a path of \emph{unbounded length}.  Consider the
problem in \figref{example3}. A separating concept is \emph{all nodes
  can reach some cycle.} By augmenting first-order logic with (least
fixed point) recursive definitions, this concept can be expressed
using the following recursive definition for $\mathit{reach}$ (which
captures reachability using at least one edge):
\begin{align*}
  \varphi \,\,\,\coloneq \quad \Let \,\, &\mathit{reach}(x,y) =_{\lfp} E(x,y)\vee \exists z.\,
                              \left(E(x,z)\wedge \mathit{reach}(z,y)\right)\,\, \\
                            \In\,\,\, &\forall x.~\exists y.\,
                              \left(\mathit{reach}(x,y)\wedge \mathit{reach}(y,y)\right)
\end{align*}

\tikzstyle{every node}=[circle, text=zgray, draw, inner sep=1.3pt, minimum width=1pt]
\tikzset{every loop/.style={in=270,out=180,looseness=5}}
\begin{figure}[H]
  \centering
  \begin{tabular}[t]{p{0.18\textwidth}>{\centering}p{0.18\textwidth}>{\centering}p{0.01\textwidth}>{\centering}p{0.18\textwidth}>{\centering}p{0.18\textwidth}}
  \scalebox{0.9}{
    \centering
    \begin{tikzpicture}[thick,scale=0.8]
      \node[shape=circle,draw=zgray,fill=zgray] (0) at (0,1) {} ;
      \node[shape=circle,draw=zgray,fill=zgray] (1) at (1,1) {} ;
      \node[shape=circle,draw=zgray,fill=zgray] (2) at (2,1) {} ;
      \node[shape=circle,draw=zgray,fill=zgray] (3) at (0,2) {} ;
      \node[shape=circle,draw=zgray,fill=zgray] (4) at (1,2) {} ;
      \node[shape=circle,draw=zgray,fill=zgray] (5) at (2,2) {} ;
      \node[shape=circle,draw=zgray,fill=zgray] (6) at (0,0) {} ;
      \node[shape=circle,draw=zgray,fill=zgray] (7) at (1,0) {} ;
      \node[shape=circle,draw=zgray,fill=zgray] (8) at (2,0) {} ;
      \node[scale=1.3,draw=none] at (1,-0.7) {\texttt{+}};

      \draw[fill=zgray] [-to] (6) edge[zgray] (0) ;
      \draw[fill=zgray] [-to] (0) edge[zgray] (1) ;
      \draw[fill=zgray] [-to] (2) edge[zgray] (8) ;
      \draw[fill=zgray] [-to] (3) edge[zgray] (4) ;
      \draw[fill=zgray] [-to] (4) edge[zgray] (5) ;
      \draw[fill=zgray] [-to] (5) edge[zgray] (2) ;
      \draw[fill=zgray] [-to] (7) edge[zgray] (6) ;
      \draw[fill=zgray] [-to] (1) edge[zgray] (7) ;
      \draw[fill=zgray] [-to] (8) edge[zgray] (7) ;
    \end{tikzpicture}
    \vspace{0.1in}
  } &
    \scalebox{0.9}{
    \centering
    \begin{tikzpicture}[thick,scale=0.8]
      \node[shape=circle,draw=zgray,fill=zgray] (0) at (0,0.5) {} ;
      \node[shape=circle,draw=zgray,fill=zgray] (2) at (0.5,1) {} ;
      \node[shape=circle,draw=zgray,fill=zgray] (4) at (1,0.5) {} ;
      \node[shape=circle,draw=zgray,fill=zgray] (6) at (0.5,0) {} ;

      \draw[fill=zgray] [-to] (0) edge[zgray, bend left=25] (2) ;
      \draw[fill=zgray] [-to] (2) edge[zgray, bend left=25] (4) ;
      \draw[fill=zgray] [-to] (4) edge[zgray, bend left=25] (6) ;
      \draw[fill=zgray] [-to] (6) edge[zgray, bend left=25] (0) ;

      \node[shape=circle,draw=zgray,fill=zgray] (8) at (-0.5,0.5) {} ;
      \node[shape=circle,draw=zgray,fill=zgray] (9) at (-1,0.5) {} ;
      \node[shape=circle,draw=zgray,fill=zgray] (10) at (0.5,1.5) {} ;
      \node[shape=circle,draw=zgray,fill=zgray] (11) at (0.5,2) {} ;
      \node[shape=circle,draw=zgray,fill=zgray] (12) at (1.5,0.5) {} ;
      \node[shape=circle,draw=zgray,fill=zgray] (13) at (2,0.5) {} ;
      \node[shape=circle,draw=zgray,fill=zgray] (14) at (0.5,-0.5) {} ;
      \node[shape=circle,draw=zgray,fill=zgray] (15) at (0.5,-1) {} ;
      \node[scale=1.3,draw=none] at (0.5,-1.5) {\texttt{+}};

      \draw[fill=zgray] [-to] (8) edge[zgray] (0) ;
      \draw[fill=zgray] [-to] (9) edge[zgray] (8) ;
      \draw[fill=zgray] [-to] (10) edge[zgray] (2) ;
      \draw[fill=zgray] [-to] (11) edge[zgray] (10) ;
      \draw[fill=zgray] [-to] (12) edge[zgray] (4) ;
      \draw[fill=zgray] [-to] (13) edge[zgray] (12) ;
      \draw[fill=zgray] [-to] (14) edge[zgray] (6) ;
      \draw[fill=zgray] [-to] (15) edge[zgray] (14) ;
    \end{tikzpicture}
    \vspace{0.1in}
      } & &
    \scalebox{0.9}{
    \centering
    \begin{tikzpicture}[thick,scale=0.8]
      \node[shape=circle,draw=zgray,fill=zgray] (0) at (1,1) {} ;
      \node[shape=circle,draw=zgray,fill=zgray] (1) at (1,0) {} ;
      \node[shape=circle,draw=zgray,fill=zgray] (2) at (2,0) {} ;
      \node[shape=circle,draw=zgray,fill=zgray] (3) at (1,-1) {} ;
      \node[shape=circle,draw=zgray,fill=zgray] (4) at (0,0) {} ;
      \node[scale=1.3,draw=none] at (1,-1.5) {\texttt{-}};

      \draw[fill=zgray] [-to] (1) edge[zgray] (0) ;
      \draw[fill=zgray] [-to] (1) edge[zgray] (2) ;
      \draw[fill=zgray] [-to] (1) edge[zgray] (3) ;
      \draw[fill=zgray] [-to] (1) edge[zgray] (4) ;
      \draw[fill=zgray] [-to] (3) edge[zgray, bend right=20] (2) ;
      \draw[fill=zgray] [-to] (0) edge[zgray, bend right=20] (4) ;
    \end{tikzpicture}
    \vspace{0.1in}
      }
   &
  \scalebox{0.9}{
    \centering
    \begin{tikzpicture}[thick,scale=0.8]
      \node[shape=circle,draw=zgray,fill=zgray] (0) at (0,1) {} ;
      \node[shape=circle,draw=zgray,fill=zgray] (1) at (0.27,1.669) {} ;
      \node[shape=circle,draw=zgray,fill=zgray] (2) at (1,2) {} ;
      \node[shape=circle,draw=zgray,fill=zgray] (8) at (1.745,2) {} ;
      \node[shape=circle,draw=zgray,fill=zgray] (3) at (2.5,1.658) {} ;
      \node[shape=circle,draw=zgray,fill=zgray] (4) at (2.745,1) {} ;
      \node[shape=circle,draw=zgray,fill=zgray] (5) at (2.525,0.39) {} ;
      \node[shape=circle,draw=zgray,fill=zgray] (6) at (1.745,0) {} ;
      \node[shape=circle,draw=zgray,fill=zgray] (9) at (1,0) {} ;
      \node[shape=circle,draw=zgray,fill=zgray] (7) at (0.265,0.337) {} ;
      \node[scale=1.3,draw=none] at (1.3725,-0.5) {\texttt{-}};

      \draw[fill=zgray] [-to] (0) edge[zgray, bend left=10] (1) ;
      \draw[fill=zgray] [-to] (1) edge[zgray, bend left=10] (2) ;
      \draw[fill=zgray,ultra thick] [to-] (2) edge[zgray, bend left=10] (8) ;
      \draw[fill=zgray] [-to] (8) edge[zgray, bend left=10] (3) ;
      \draw[fill=zgray] [-to] (3) edge[zgray, bend left=10] (4) ;
      \draw[fill=zgray] [to-] (5) edge[zgray, bend right=10] (4) ;
      \draw[fill=zgray] [-to] (5) edge[zgray, bend left=10] (6) ;
      \draw[fill=zgray] [-to] (6) edge[zgray, bend left=10] (9) ;
      \draw[fill=zgray] [-to] (9) edge[zgray, bend left=10] (7) ;
      \draw[fill=zgray] [-to] (7) edge[zgray, bend left=10] (0) ;
    \end{tikzpicture}
    \vspace{0.1in}
      }
  \end{tabular}
  \vspace{-0.1in}
  \caption{Find a sentence in first-order logic with least fixed point
    definitions that is true on \texttt{+} structures and false on \texttt{-}
    structures and that does not depend on the sizes of the
    structures.  }
  \figlabel{example3}
\end{figure}
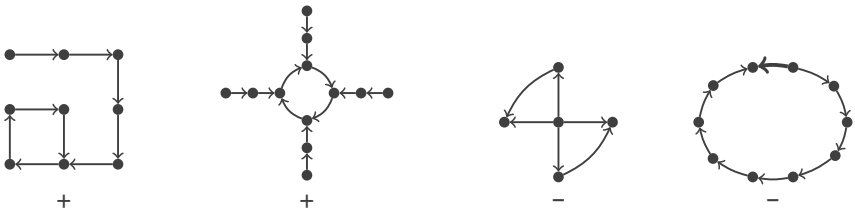
Recursive definitions dramatically increase the expressivity of
first-order logic. As studied in finite model theory, such logics
encompass (over structures equipped with a linear order on the domain)
\emph{all} polynomial-time computable properties, i.e., the class
$\mathbb{P}$~\cite{vardi82,immerman82,libkinmodeltheory}. We consider
learning in a finite-variable version of first-order logic with least
fixed points that captures all properties computable in time $n^k$
using $\Oo(k)$ variables. We note here a connection to the problem of
synthesizing programs that are syntactically restricted in order to
guarantee a specific \emph{implicit
  complexity}~\cite{implicit-complextiy-dellago}, say polynomial time;
we leave an exploration of this connection to future work.

\subsection{Example 4: Term Synthesis and Program Synthesis}
\label{sec:example-4}
In addition to the separability problem illustrated in the previous
examples, we also study the problem of \emph{term synthesis}. In the
term synthesis problem we aim to synthesize a term that evaluates to a
specific element of the domain for each structure in a set of input
structures. Specifically, given a set of structures, each with an
interpretation for constants $\mathit{in}_1, \ldots, \mathit{in}_d,$
and $\mathit{out}$, we want to construct a term $t$ in the language of
a grammar $G$ such that $t$ has the same interpretation as
$\mathit{out}$ in each structure. (Note the structures are not labeled
in this problem.)  The term synthesis problem, especially in the
presence of recursive function definitions, resembles \emph{functional
  program synthesis}.

Consider the problem of merging two sorted lists. We can model this
setting with structures that represent finite prefixes of an abstract
datatype for lists over a linearly ordered finite domain
$\{a_1, a_2, \ldots, a_n \}$ (with ordering $<$). \figref{example4}
(top) depicts a portion of one such structure and its operations.  We
can model input-output tuples for the desired merge operation using an
interpretation of constants $in_1,in_2,$ and $out$, as depicted in the
bottom of the figure.

\tikzstyle{every node}=[rectangle, text=black, draw, inner sep=2pt, minimum width=4pt]
\tikzset{every loop/.style={in=270,out=180,looseness=5}}
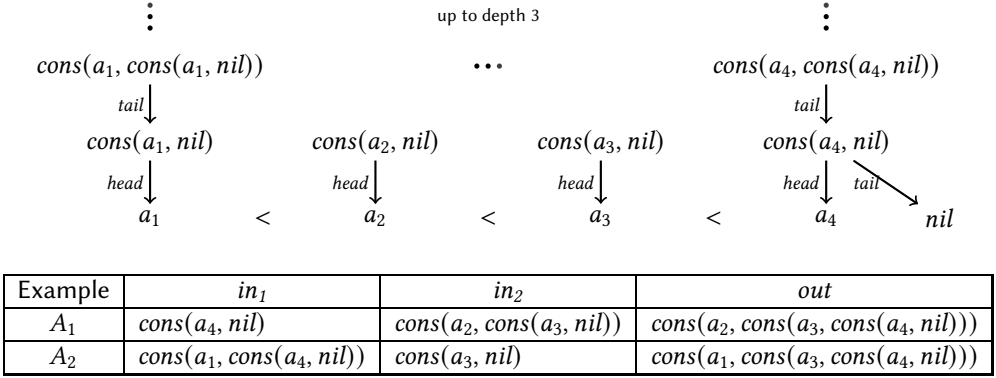
\begin{figure}[H]
  \centering
  \begin{minipage}[c]{\linewidth}
    \centering
    \begin{tikzpicture}[thick,scale=0.5]
      \node[draw=none] (nil) at (21,0)
      {$\nil$} ;
      \node[draw=none] (a1) at (0,0)
      {$a_1$} ;
      \node[draw=none] (cons1nil) at (0,2)
      {$\cons(a_1,\nil)$} ;
      \node[draw=none] (lt1) at (3,0)
      {$<$} ;
      \node[draw=none] (a2) at (6,0)
      {$a_2$} ;
      \node[draw=none] (cons2nil) at (6,2)
      {$\cons(a_2,\nil)$} ;
      \node[draw=none] (lt2) at (9,0)
      {$<$} ;
      \node[draw=none] (a3) at (12,0)
      {$a_3$} ;
      \node[draw=none] (cons3nil) at (12,2)
      {$\cons(a_3,\nil)$} ;
      \node[draw=none] (lt3) at (15,0)
      {$<$} ;
      \node[draw=none] (a4) at (18,0)
      {$a_4$} ;
      \node[draw=none] (cons4nil) at (18,2)
      {$\cons(a_4,\nil)$} ;
      \node[draw=none] (cons11nil) at (0,4)
      {$\cons(a_1,\cons(a_1,\nil))$} ;
      \node[draw=none] (cons44nil) at (18,4)
      {$\cons(a_4,\cons(a_4,\nil))$} ;

      \draw[fill=zgray, to-] (a1) -- (cons1nil)
      node[midway,left,draw=none] {{\scriptsize $\car$}} ;

      \draw[fill=zgray, to-] (a2) -- (cons2nil)
      node[midway,left,draw=none] {{\scriptsize $\car$}} ;

      \draw[fill=zgray, to-] (a3) -- (cons3nil)
      node[midway,left,draw=none] {{\scriptsize $\car$}} ;

      \draw[fill=zgray, to-] (a4) -- (cons4nil)
      node[midway,left,draw=none] {{\scriptsize $\car$}} ;

      \draw[fill=zgray, to-] (nil) -- (cons4nil)
      node[midway,left,draw=none] {{\scriptsize $\cdr$}} ;

      \draw[fill=zgray, to-] (cons1nil) -- (cons11nil)
      node[midway,left,draw=none] {{\scriptsize $\cdr$}} ;

      \draw[fill=zgray, to-] (cons4nil) -- (cons44nil)
      node[midway,left,draw=none] {{\scriptsize $\cdr$}} ;

      \node[shape=circle,draw=black,fill=zgray,inner sep=0.4pt,minimum width=0.2pt]
      (ellipsesLeft1) at (0,5) {} ;
      \node[shape=circle,draw=black,fill=zgray,inner sep=0.4pt,minimum width=0.2pt]
      (ellipsesLeft2) at (0,5.3) {} ;
      \node[shape=circle,draw=zgray,fill=zgray,inner sep=0.4pt,minimum width=0.2pt]
      (ellipsesLeft3) at (0,5.6) {} ;

      \node[shape=circle,draw=black,fill=zgray,inner sep=0.4pt,minimum width=0.2pt]
      (ellipsesRight1) at (18,5) {} ;
      \node[shape=circle,draw=black,fill=zgray,inner sep=0.4pt,minimum width=0.2pt]
      (ellipsesRight2) at (18,5.3) {} ;
      \node[shape=circle,draw=zgray,fill=zgray,inner sep=0.4pt,minimum width=0.2pt]
      (ellipsesRight3) at (18,5.6) {} ;

      \node[shape=circle,draw=black,fill=zgray,inner sep=0.4pt,minimum width=0.2pt]
      (ellipsesMid1) at (8.7,4) {} ;
      \node[shape=circle,draw=black,fill=zgray,inner sep=0.4pt,minimum width=0.2pt]
      (ellipsesMid2) at (9,4) {} ;
      \node[shape=circle,draw=zgray,fill=zgray,inner sep=0.4pt,minimum width=0.2pt]
      (ellipsesMid3) at (9.3,4) {} ;

      \node[draw=none]
      (ellipsesMid2) at (9,5.3) {\scriptsize \textsf{up to depth 3}} ;







    \end{tikzpicture}
    \vspace{0.1in}
  \end{minipage}
  \begin{minipage}[c]{\linewidth}
    \vspace{0.1in}
    \centering
    \begin{tabular}[t]{|l|l|l|l|} \hline
      $\mathsf{Example}$ & \multicolumn{1}{c|}{$\mathit{in_1}$} & \multicolumn{1}{c|}{$\mathit{in_2}$} & \multicolumn{1}{c|}{$\mathit{out}$} \\\hline
      \multicolumn{1}{|c|}{$A_1$} & $\cons(a_4,\nil)$ &
                                             $\cons(a_2,\cons(a_3,\nil))$ & $\cons(a_2,\cons(a_3,\cons(a_4,\nil)))$ \\\hline
      \multicolumn{1}{|c|}{$A_2$} & $\cons(a_1,\cons(a_4,\nil))$ &
                                             $\cons(a_3,\nil)$ &
                                                                 $\cons(a_1,\cons(a_3,\cons(a_4,\nil)))$ \\ \hline
    \end{tabular}

  \end{minipage}
  \vspace{0.1in}
  \caption{(Top) Partial picture of a structure $A_1$ that encodes a
    finite prefix of a datatype for lists over an ordered domain, with
    terms bounded to depth 3. (Bottom) Input-output examples for
    $\mathit{merge}$. The goal is to find a closed term in first-order
    logic with least fixed point relations and recursive functions
    that evaluates to $out$ on structures $A_1$ (top) and $A_2$ (not
    shown). }
  \figlabel{example4}
\end{figure}

\vspace{-0.1in}
One possible solution is the following term $t$ that defines a
recursive function $\mathit{merge}$ and applies it to the inputs
$\mathit{in_1},\mathit{in_2}$:
\begin{align*}
  t \,\,\,\coloneq \quad &\Let \,\, \mathit{merge}(x,y) =_{\lfp} \,\,
                           \iteterm(x=\nil,y, \iteterm(y=\nil,x, (\iteterm(\car(y) >
                           \car(x), \\
                         &\quad\quad\quad
                           \quad\quad\quad\quad\quad\quad \cons(\car(x),\mathit{merge}(\cdr(x),y)), \\
                         &\quad\quad\quad
                           \quad\quad\quad\quad\quad\quad \cons(\car(y),\mathit{merge}(x,\cdr(y))))))) \\
                         & \In\,\,\, \mathit{merge}(\mathit{in_1},\mathit{in_2})
\end{align*}


\section{Background}
\seclabel{prelim}

We begin with some preliminary notions from logic as well as the
concepts of \emph{term}, \emph{tree}, and \emph{regular tree grammar},
which we will need for our definitions of various automata.

\subsection{Logic}
\label{sec:logic}

\subsubsection{Structures and Signatures}
\label{sec:struct-sign}
A \emph{first-order signature}, or simply \emph{signature}, is a set
$\tau$ of sets of relation symbols $\{R_1,R_2,\ldots\}$, function
symbols $\{f_1,f_2,\ldots\}$, and constant symbols
$\{c_1,c_2,\ldots\}$. Each symbol $s$ has an associated arity, denoted
$\arity{s}\in\Nat$. The meaning of symbols in a signature depends on a
\emph{$\tau$-structure}, which is a tuple
$A=\la \dom(A),
R_1^A,\ldots,R_a^A,f^A_1,\ldots,f^A_b,c^A_1,\ldots,c^A_c\ra$. The
domain $\dom(A)$ is a set, each $R_i^A$ is a relation on the domain,
i.e., $R_i^A\subseteq \dom(A)^{\arity{R_i}}$, each $f_j^A$ is a total
function on the domain, i.e.,
$f_j^A : \dom(A)^{\arity{f_j}}\rightarrow\dom(A)$, and each constant
$c$ denotes an element $c^A\in \dom(A)$. For simplicity, we model
constants as nullary functions. Each problem addressed in this work
involves \emph{finite} structures, i.e., those for which
$|\dom(A)|\in\Nat$. Thus \emph{structure} will always mean
\emph{finite structure}. We omit $\tau$ and write \emph{structure}
whenever $\tau$ can be understood from context or is unimportant. We
use $A$ to denote an arbitrary structure.

\subsubsection{First-Order Logic}
\label{sec:first-order-logic}
Though the technique presented in this work is highly versatile, we
will focus the majority of our presentation on variants and extensions
of first-order logic. As a starting point, we consider first-order
logic extended with an \ite~ term. Syntax for this logic, denoted
$\FO$, is given in~\figref{FO}.  The semantics of the usual $\FO$
formulas and terms is standard. We denote the interpretation of a term
$t$ in a structure $A$ and variable assignment $\gamma$ as
$t^{A,\gamma}$. The interpretation of the \ite~ term in $A,\gamma,$
is:
\vspace{-0.1in}
\begin{align*}
  \IfThenElse{\varphi}{t_1}{t_2}^{A,\gamma} =
  \begin{cases} t_1^{A,\gamma} &
    A,\gamma \models\varphi \\ t_2^{A,\gamma} & \text{otherwise}
  \end{cases}
\end{align*}
We use $\Ll$ to refer to an arbitrary logic, and occasionally, if we
want to emphasize the signature we write $\Ll(\tau)$ for a logic $\Ll$
over $\tau$. See ~\cite{enderton} for syntax, semantics, and basic
results in first-order logic.

\begin{figure}
  \centering
  \begin{align*}
    \varphi \Coloneqq R(\many{t}) \mid \varphi \vee \varphi \mid \varphi \wedge \varphi
    \mid \neg \varphi \mid \exists x.\varphi \mid \forall x.\varphi \quad\quad
    t \Coloneqq x \mid c \mid f(\many{t}) \mid
            \IfThenElse{\varphi}{t}{t'}
  \end{align*}
  \caption{Grammar for first-order logic with if-then-else terms,
    denoted $\FO$. }
  \figlabel{FO}
\end{figure}

\subsection{Terms and Trees}
\label{sec:term-algebra}
Rather than working with strings, it will be simpler to instead
consider logical formulas as finite ordered ranked trees, sometimes
called \emph{terms}. Intuitively, to build terms we use symbols from a
finite \emph{ranked alphabet}, that is, a set of symbols with
corresponding arities. We use $T_\alphabet(X)$ to denote the set of
terms over a ranked alphabet $\alphabet$ augmented with nullary
symbols $X$ (with $X$ disjoint from $\alphabet$). When $X=\emptyset$
we just write $T_\alphabet$.

It will be convenient to also use the language of ordered trees. An
\emph{ordered tree} $\tree$ over a label set $W$ is a partial function
$\tree : \Nat^*\rightarrow W$ defined on
$\treepos(\tree)\subseteq \Nat^*$, a prefix-closed set of positions
containing a root $\epsilon\in \treepos(\tree)$. In this view, terms
are simply ordered trees whose labels respect ranks, that is, ordered
trees subject to the following requirement: if $a\in \alphabet$ with
$\arity{a}=n$ and for some $x\in \treepos(\tree)$ we have
$\tree(x) = a$, then
$\{j\in\Nat \,\mid\, x \cdot j\in \treepos(\tree)\} = \{1,\ldots,
n\}$. We will refer to ordered (ranked) trees as simply \emph{trees}
to avoid confusion with the usual syntactic category of \emph{logical
  terms}.

\subsection{Finite Variable Logics to Trees}
\label{sec:finite-vari-logics}

When can the formulas of a logic be represented as trees over a finite
alphabet? We probably must have a finite signature, as well as syntax
formation rules that take a finite number of subformulas. For any
logic $\Ll$ and signature $\tau$ that meets these requirements, if we
bound the number of variables that can appear in any formula, then we
can define a \emph{finite} ranked alphabet $\alphabet_{\Ll(\tau)}$
such that any formula $\varphi\in\Ll(\tau)$ has at least one
corresponding tree $t\in T_{\alphabet_{\Ll(\tau)}}$. For example,
consider a variant of $\FO$ restricted to the $k$ variables in
$V = \{x_1,\ldots, x_k\}$, which we denote $\FOk$\footnote{We overload
  notation, using $\FO(\tau)$ for $\FO$ over signature $\tau$ and
  $\FOk$ for $\FO$ with $k$ variables and an unspecified signature. If
  the two notations are both needed at once we put the signature last,
  e.g., $\FO(k)(\tau)$ is $\FO$ with $k$ variables over
  $\tau$. }. With superscripts for arity, the ranked alphabet looks as
follows:
\begin{align*}
  \alphabet_{\FOk} = \left\{R^{\arity{R}} \bigm\vert R\in\tau
   \right\} \cup \left\{f^{\arity{f}} \bigm\vert f\in\tau
   \right\} \cup \left\{\iteterm^{3}\right\} \cup \left\{\andsymb^2,
  \orsymb^2, \negsymb^1\right\} \cup \left\{\allsymb{x}^1,
  \existsymb{x}^1, x^0 \mid x\in V \right\}
\end{align*}
We sometimes drop the subscript for the underlying logic and just use
$\alphabet$ to refer to finite ranked alphabets of this kind.

\subsection{Regular Tree Grammars}
\label{sec:regul-tree-gramm}
A regular tree grammar (RTG) is a tuple
$G = \la N, \alphabet, S, P\ra$, where $N$ is a finite set of
nonterminal symbols, $\alphabet$ is a finite ranked alphabet, $S\in N$
is the axiom nonterminal, and $P$ is a set of rewrite rules of the
form $B \rightarrow t$, where $B\in N$ and $t\in T_\alphabet(N)$. The
language $L(G)$ of $G$ is the set of trees
$\left\{t\in T_\alphabet\mid S\Rightarrow^* t\right\}$, where
$t \Rightarrow t'$ holds whenever there is a context $C$ and tree
$t''\in T_\alphabet(N)$ such that $t=C[B], t'=C[t'']$ and
$B\rightarrow t''\in P$. These are standard notions; see~\cite{tata,automata-logics-games}
for details.

Given a logic $\Ll$, we consider RTGs \emph{over $\Ll$}. If
$\alphabet_\Ll$ is a finite ranked alphabet for $\Ll$, then an RTG
over $\Ll$ is of the form $G = \la N,\alphabet_\Ll,S,P\ra$ for some
$N,S,P$. In~\figref{rtg} we give an example RTG over
$\FO(k)(\mathsf{graph})$ with $k=2$, i.e., first-order logic with
variables $V=\{x,y\}$ over a signature $\mathsf{graph}$ consisting of
a single binary relation symbol $E$. When we refer to sentences or
formulas in the remainder of this paper we mean their corresponding
trees in a suitable RTG, and when we refer to a \emph{grammar} we mean
an RTG. When there are multiple ways to represent a given formula as a
tree, we pick one arbitrarily.
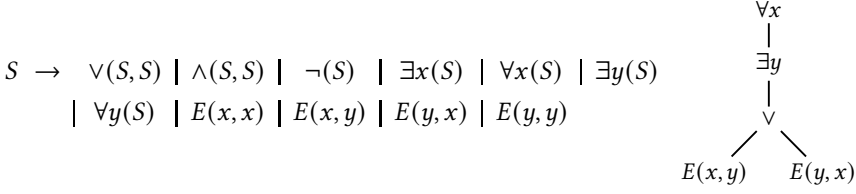
\begin{figure}
  \hfill\begin{minipage}[c]{0.6\linewidth}
    \centering
    \renewcommand{\arraystretch}{0.5}
    \begin{tabu}[t]{c|[0.6pt]c|[0.6pt]c|[0.6pt]c|[0.6pt]c|[0.6pt]c|[0.6pt]c}
      \multicolumn{1}{c}{$S\,\,\rightarrow$} & $\orsymb(S,S)$ &
      $\andsymb(S,S)$
      &$\negsymb(S)$ & $\existsymb{x}{(S)}$ & $\allsymb{x}{(S)}$ &
      $\existsymb{y}{(S)}$ \\
      \multicolumn{1}{c}{} \\
      & $\allsymb{y}{(S)}$ &
      $E(x,x)$ & $E(x,y)$ & $E(y,x)$ & \multicolumn{1}{c}{$E(y,y)$} \\
    \end{tabu}
  \end{minipage}\hfill
   \begin{minipage}[c]{0.3\linewidth}
     \scalebox{0.9}{
       \begin{tikzpicture}[thick,scale=0.8]
         \node[draw=none] (x) at (0,0) {$\forall x$} ;
         \node[draw=none] (y) at (0,-1) {$\exists y$} ;
         \node[draw=none] (or) at (0,-2) {$\vee$} ;
         \node[draw=none] (exy) at (-1,-3) {$E(x,y)$} ;
         \node[draw=none] (eyx) at (1,-3) {$E(y,x)$} ;
         \draw [-] (x) edge[black] (y) ;
         \draw [-] (y) edge[black] (or) ;
         \draw [-] (or) edge[black] (exy) ;
         \draw [-] (or) edge[black] (eyx) ;
       \end{tikzpicture}
       \vspace{0.1in}
     }
   \end{minipage}
   \caption{(Left) Production rules from the set $P$ for a regular
     tree grammar $G = \la\{S\}, \alphabet_\Ll, S, P\ra$, where $\Ll$
     is $\FO(k)(\graph)$ with $k=2$, and (Right) a tree in $L(G)$ for
     the sentence $\forall x.\exists y.\, (E(x,y)\vee E(y,x))$. }
   \figlabel{rtg}
\end{figure}


\subsection{Alternating Tree Automata}
\label{sec:ata}
As we will show, the formalism of \emph{alternating tree automata}
yields an elegant technique for evaluating expressions on a fixed
structure. We summarize the relevant ideas from automata theory.

An alternating tree automaton (ATA) over $\kappa$-ary trees is a tuple
$\aut = \la Q,\Sigma,I,\delta \ra$, where $Q$ is a finite set of
states, $\Sigma$ is a finite ranked alphabet, $I\subseteq Q$ is a set
of initial states, and the transition function has the form
$\delta : Q\times\Sigma\rightarrow \Bb^+(Q\times\{1,\ldots,\kappa\})$,
where $\Bb^+(X)$ denotes the set of positive propositional formulas
over atoms from a set $X$. For any $(q,a)\in Q\,\times\,\Sigma$ we
require $\delta(q,a)\in \Bb^+(Q\times\{1,\ldots, \arity{a}\})$. For
example, if $f\in\Sigma$ and $\arity{f}=2$, we might have
$\delta(q,f) = (q_1,1) \wedge (q_2,2) \vee (q_1',1)\wedge (q_2',2)$.
This transition stipulates that, when reading the symbol $f$ in state
$q$, the automaton must \emph{either} successfully continue from the
left child in state $q_1$ and from the right child in state $q_2$
\emph{or} it must successfully continue from the left child in state
$q_1'$ and from the right child in state $q_2'$. This, in fact, is
already expressible as a \emph{nondeterministic tree automaton}
transition. A nondeterministic tree automaton can be viewed as an ATA
whose transition formulas are in disjunctive normal form, where each
conjunctive subformula refers to each child at most once. What
alternation buys is a transition like, e.g.,
$\delta(q,f) = (q_1,1) \wedge (q_2,2) \wedge (q_1',1)\wedge (q_2',2)$,
in which there are multiple distinct conditions placed on a single
child.

The \emph{language} $L(\aut)$ of an ATA $\aut$ is the set of trees
that it accepts. This set is defined with respect to a \emph{run},
which captures the idea of a pass over an input tree that succeeds in
satisfying the conditions stipulated by the transition function. For
the automata in this work, some transitions will use the formula
$\fals$, and trees in the language of such an automaton can be thought
of as those which are able to satisfy the transition formulas in such
a way that they never are forced to satisfy $\fals$, which is
impossible.

A \emph{run} of an ATA $\aut = \la Q,\Sigma,I,\delta\ra$ on an input
$t\in T_\Sigma$ is an ordered tree $\tree$ over the label set
$Q\times \treepos(t)$ satisfying the following two conditions:
\begin{itemize}
\item $\tree(\epsilon) = (q_i,\epsilon)$ for some state $q_i\in I$
\item Let $n\in \treepos(\tree)$. If $\tree(n) = (q,x)$ with
  $t(x) = a$, then there exists
  $S = \{(q_{1},i_1),\ldots,(q_{l},i_l)\}\subseteq
  Q\times\{1,\ldots,\arity{a}\}$ such that $S\models \delta(q,a)$ and
  $\tree(n\cdot j) = (q_j,x\cdot i_j)$ for $1\le j\le l$.
\end{itemize}
An ATA $\aut$ \emph{accepts} a tree $t$ if it has a run on $t$. Note
the structure of a run $\rho$ can be different from that of the input
$t$, since it records how the automaton satisfies the transition
function, which can involve going to several states for any given
child. This is important for the complexity of emptiness checking for
ATAs, because it means that mere reachability of states is not enough
to verify a transition can be taken: for some transitions one must
also verify that reachability of certain states is \emph{witnessed by
  the same tree}. We note that (1) an ATA can be converted to a
nondeterministic tree automaton with the same language in exponential
time (incurring an exponential increase in states) and (2) emptiness
for nondeterministic tree automata is decidable in linear time.

We will use $L(\aut,q)$ to denote the language of an automaton $\aut$
when we view $q\in Q$ as an initial state (thus
$L(\aut)=\cup_{q\in I}L(\aut,q)$). We refer the reader to~\cite{tata}
for details and standard results about tree automata.


\section{Realizability and Synthesis Problems}
\seclabel{problems}

In this section we define three exact learning problems that are
parameterized by a logic $\Ll$. The first problem involves separating
a set of labeled structures using a sentence in $\Ll$ . The second
problem involves finding a formula in $\Ll$ that exactly defines a
given set of tuples over the domain of a structure. The third problem
involves finding a term in $\Ll$ that obtains specified domain values
in given structures. For each of these problems, we always assume a
fixed and finite signature $\tau$.

\begin{table}
  \caption{Summary of main results for a fixed signature and
    fixed arities of relations and functions. }
  \centering
  \begin{tabular}{ |l|c|c|c| }
    \hline
    Problem & Parameters & Time complexity
    & \begin{tabular}{@{}c@{}}
        Combined complexity \\
        {\small(fixed variables)}
      \end{tabular}
    \\ \hline
    \begin{tabular}{@{}l@{}}
      $\FO$ separability
    \end{tabular}
          & \multirow{3}{*}[-0.07in]{\scalebox{0.9}{\begin{tabular}{@{}l@{}}
               $m$ input structures \\
               $n$ max structure size \\
               $k$ first-order variables \\
             \end{tabular}}}
          & \multirow{3}{*}[-0.15in]{\begin{tabular}{@{}c@{}}
              $\Oo\left(2^{\mathit{poly}\left(m n^k\right)}|G|\right)$
             \end{tabular}}
          & \multirow{5}{*}[-0.3in]{\begin{tabular}{@{}c@{}}
              $\EXP$-complete \\ in $\,m n + |G|$
            \end{tabular}}
    \\
    \cline{1-1}
    \begin{tabular}{@{}@{}l}
      $\FO$ queries
    \end{tabular} & & &
    \\
    \cline{1-1}
    \begin{tabular}{@{}l@{}}
      $\FO$ term \\ synthesis
    \end{tabular} & & &
    \\ \cline{1-3} 
    \begin{tabular}{@{}l@{}}
      $\FOLFP$ \\ separability
      \end{tabular}
          & \scalebox{0.9}{\begin{tabular}{@{}l@{}}
                 $\FO$ parameters \emph{and} \\
                    $k'$ relation variables \\
                  \end{tabular}}
            & \begin{tabular}{@{}c@{}}
                $\Oo\left(2^{\mathit{poly}\left(m n^{k}
  k'\right)}|\grammar|\right)$\\
              \end{tabular}
            &
    \\ \cline{1-3} 
    \begin{tabular}{@{}l@{}}
      $\FOTERM$ \\ term synthesis
      \end{tabular}
          & \scalebox{0.9}{\begin{tabular}{@{}l@{}}
                 $\FO$ parameters \emph{and} \\
                    $k_1$ relation variables \\
                    $k_2$ function variables \\
                  \end{tabular}}
          & \begin{tabular}{@{}c@{}}
              $\Oo\left(2^{\mathit{poly}\left(m
              n^{k}(k_1+k_2)\right)}|G|\right)$ \\
              \end{tabular}
            &
    \\ \hline
  \end{tabular}
  \label{tab:results}
\end{table}

The first problem, \sepsynth{$\Ll$}, is defined in
\probref{sepsynth}. Given positively- and negatively-labeled
structures and a grammar $G$ over $\Ll$, the problem is to synthesize
a sentence in $G$ such that all positive structures make the sentence
true and all negative structures make it false, or declare no such
sentence exists. Sometimes we refer to this as the \emph{separability}
problem.

\begin{algorithm}
  \NoCaptionOfAlgo
  \caption{\textbf{Problem 1:} \sepsynth{$\Ll$}}
  \problabel{sepsynth}
  \BlankLine
  \KwIn{$\left\la \mathit{Pos}=\{A_1, \ldots, A_{m_1}\}, \mathit{Neg}=\{B_1, \ldots, B_{m_2}\}, G \right\ra$ where}

  \myinput{$A_i, B_j$ are $\tau$-structures}

  \myinput{$G$ an RTG over $\alphabet_{\Ll}$}

  \KwOut{$\varphi\in L(G) \,\, s.t.$ for all $A_i\in \mathit{Pos},\,
    A_i\models\varphi$, and for all $B_j\in \mathit{Neg},\,B_j\not\models\varphi$}

  \myoutput{Or ``No'' if no such $\varphi$ exists}
\end{algorithm}

The second problem involves synthesizing $r$-ary \emph{queries}. A
$r$-ary query for a logic $\Ll$ is a formula
$\varphi(x_1,\ldots,x_r)\in\Ll$ that has exactly $r$ distinct free
variables, all first-order. The \emph{answer set} for a $r$-ary query
$\varphi$ in a structure $A$ is the precise set of tuples
$\mathit{Ans} = \{\, \many{a}\in \dom(A)^r \,\mid\, A\models
\varphi(\many{a}) \,\}$ that make the query true in the structure. For
example, consider a ``family relationships'' domain with two
structures $A_1$ and $A_2$. In $A_1$ there are domain elements
$\mathit{Sue}$ and $\mathit{Bob}$ and the relationship
$\mathsf{Mother}(\mathit{Sue},\mathit{Bob})$, and in $A_2$ there are
elements $\mathit{Maria},\mathit{Tom},$ and $\mathit{Anne}$ and the
relationships $\mathsf{Mother}(\mathit{Maria},\mathit{Tom})$ and
$\mathsf{Mother}(\mathit{Maria},\mathit{Anne})$. Suppose the answer
sets are $\mathit{Ans_1} = \{\mathit{Sue}\}$ and
$\mathit{Ans_2} = \{\mathit{Maria}\}$. Then one possible solution is
the query $\varphi(x) \coloneq \exists y. \mathsf{Mother}(x,y)$.

We call this second problem \querysepsynth{$\Ll$}, which is defined
formally in \probref{querysepsynth}. Given a grammar $G$ over $\Ll$
and a set of pairs, where each pair is a structure and an answer set,
synthesize a query $\varphi$ in $G$ such that $\varphi$ precisely
defines the given answer set in each structure, or declare no such
$\varphi$ exists.

\begin{algorithm}
  \NoCaptionOfAlgo
  \caption{\textbf{Problem 2:} \querysepsynth{$\Ll$}}
  \problabel{querysepsynth}
  \BlankLine
  \KwIn{$\left\la \{\la A_1, \mathit{Ans_1}\ra,\ldots, \la A_{m},\mathit{Ans_{m}}\ra\}, G \right\ra$ where}

  \myinput{$A_i$ are $\tau$-structures}

  \myinput{$\mathit{Ans_i}\subseteq\dom(A_i)^r$}

  \myinput{$G$ an RTG over $\alphabet_{\Ll}$}

  \KwOut{$\varphi(x_1,\ldots,x_r)\in L(G) \,\, s.t. \,\, \left\{\many{a}\in\dom(A_i)^r \mid A_i\models\varphi(\many{a})\right\} = \mathit{Ans_i}$ for all $i\in \left[m\right]$}
  \myoutput{Or ``No'' if no such $\varphi$ exists}
\end{algorithm}

The third problem, \termsynth{$\Ll$}, is defined in
\probref{termsynth}. The input is a grammar $G$ and a set of
(unlabeled) structures $\{A_1,\ldots,A_m\}$. Each structure $A_i$
interprets constants $\mathit{in_1},\ldots,\mathit{in_d}$, and
$\mathit{out}$, where $\mathit{out}$ is the \emph{target element} in
the domain of each structure. The goal is to synthesize a term $t$
from $G$ (which precludes using $\mathit{out}$) such that
$A_i\models (t = \mathit{out})$ for each $i$.

\begin{algorithm}
  \NoCaptionOfAlgo
  \caption{\textbf{Problem 3:} \termsynth{$\Ll$}}
  \problabel{termsynth}
  \BlankLine
  \KwIn{$\left\la \{A_1, \ldots, A_m\}, G \right\ra$ where}

  \myinput{$A_i$ are $\tau$-structures}

  \myinput{$G$ an RTG over $\alphabet_{\Ll}$}

  \KwOut{$t\in L(G) \,\, s.t.\,\, A_i\models (t = \mathit{out})$ for all $i\in [m]$}

  \myoutput{Or ``No'' if no such $t$ exists}
\end{algorithm}

In \tableref{results} we highlight our main results for these three
problems instantiated with various logics. Note that the upper bounds
on time complexity assume a fixed signature, and in particular, fixed
arities of symbols. The remainder of the paper lays out our general
automata-theoretic solution.


\section{Solving Realizability and Synthesis for First-Order Logic}
\label{sec:upper-bounds}

In this section we describe our general technique by instantiating it
on the separability problem for the logic $\FOk$
(\csecref{foksepsynth}). We then show how to adapt the solution to
solve query synthesis for the same logic
(\csecref{fokquerysepsynth}). Term synthesis is covered in
\csecref{term-synthesis}.

\subsection{Separator Realizability and Synthesis in First-Order Logic}
\label{sec:foksepsynth}

Consider separability for $\FOk$ over an arbitrary signature. We are
given a grammar $G$ and sets of positive and negative structures
$\Pos$ and $\Neg$. The main idea is to build an alternating tree
automaton that accepts the parse trees of \emph{all} formulas that
separate $\Pos$ and $\Neg$. This automaton itself is constructed as
the product of automata $\aut_M$, one for each structure
$M\in \Pos\uplus \Neg$. If $M\in \Pos$, then $\aut_M$ accepts all
formulas that are true on $M$. If $M\in \Neg$, then $\aut_M$ accepts
all formulas that are false on $M$. Clearly, the intersection of these
automata gives the desired automaton $\aut_\cap$ that accepts formulas
which separate $\Pos$ and $\Neg$. In \csecref{constructions}, we give
the main construction of $\aut_M$ for each $M\in \Pos\uplus \Neg$,
which involves evaluating a given input formula on a fixed structure
$M$.

We build another tree automaton $\aut_\grammar$
(\csecref{grammar-automaton}) that accepts precisely the formulas from
$G$, and finally we construct an automaton accepting the intersection
of languages for $\aut_\cap$ and $\aut_\grammar$
(\csecref{sepsynth-decision-procedure}). Checking emptiness of this
automaton solves the realizability problem and, when the language is
nonempty, finding a member of the language solves the synthesis
problem.

\subsubsection{Automaton for Evaluating First-Order Logic Formulas}
\label{sec:constructions}

We now show how to construct a tree automaton that accepts the set of
sentences in $\FOk$ that are true in a given structure $A$. For
clarity, we present an automaton for a slightly simpler version of
$\FOk$ over an arbitrary relational signature $\tau$ and without an
\ite~ term (thus the only terms are variables). The ranked alphabet
for this simplification over $\tau=\la R_1,\ldots,R_s\ra$ is:
\begin{align*}
  \alphabet'_{\FOk} = \left\{R_i(\many{x})^{0} \bigm\vert R_i\in\tau,
  \many{x}\in V^{\arity{R_i}} \right\} \cup \left\{\andsymb^2, \orsymb^2, \negsymb^1\right\} \cup \left\{\allsymb{x}^1, \existsymb{x}^1 \mid x\in V \right\}
\end{align*}
Note that each atomic formula over variables $V$ becomes a nullary
symbol. Handling the full gamut of terms in $\FO$ from \figref{FO} is
straightforward but tedious, so we omit the details. Following the
simpler construction, we give the high-level idea for the full
version.

Fix a $\tau$-structure $A$ with $|\dom(A)|=n$. We define an ATA
$\aut_A=\la Q,\alphabet'_{\FOk},I,\delta \ra$ whose language is the
set of trees over $\alphabet'_{\FOk}$ corresponding to sentences that
are true in the structure $A$. Each component is discussed below.

\paragraph{\textsf{States.}} The states of $\aut_A$ are partial
assignments from variables $V=\{x_1,\ldots, x_k\}$ to the domain
$\dom(A)$. We denote the set of partial assignments by
$\mathit{Assign}\coloneq V\rightharpoonup \dom(A)$, and we use
$\gamma$ to range over $\mathit{Assign}$. States of the automaton keep
track of assignments that accrue when the automaton reads
quantification symbols. The crucial idea is that for each syntax
formation rule, we can express the conditions under which the formula
is true with the current assignment as a positive Boolean formula over
assignments and subformulas. The only hiccup is that the automaton
needs to keep track of whether or not a formula should be satisfied or
\emph{not} satisfied, which is dictated by occurrences of negation. We
need a single bit for this, and so the state space increases by a
factor of two. A state $\gamma\in\mathit{Assign}$ can be marked
$\tilde{\gamma}$ to indicate that under assignment $\gamma$ the input
formula should \emph{not} be true. We use
$\mathit{Dual}(X)\triangleq\{x,\tilde{x}\,\mid\,x\in X\}$ to denote a
set $X$ together with marked copies of its elements. With this
notation, the state set for our automaton is
$Q\coloneq \mathit{Dual}(\mathit{Assign})$, and $|Q| = \Oo(n^k)$. For
a given $\gamma$ we abuse notation and treat $\gamma$ as a set of
variable-binding pairs, for instance, $\{x\mapsto a_1, y\mapsto a_2\}$
and $\emptyset$ denote assignments in this way. We use
$\update{\gamma}{x}{a}$ to denote the assignment that is identical to
$\gamma$ except it maps $x$ to $a$. We write
$\gamma(\many{x})\downarrow$ to denote that $\gamma$ is defined on
each $x_i\in\many{x}$ and $\gamma(\many{x})$ to denote the tuple of
elements obtained by applying $\gamma$ to $\many{x}$.

\paragraph{\textsf{Initial states.}} There is only one initial state,
namely, the one that assigns no variables: $I = \{\emptyset\}$.

\paragraph{\textsf{Transitions.}} To define the transition function,
for each assignment $\gamma\in \mathit{Assign}$ and each symbol
$a\in\alphabet'_{\FOk}$, we give a propositional formula that
naturally mimics the semantics of first-order logic. The intuition is
that, from a state $\gamma\in\mathit{Assign}$, the automaton accepts
every formula that is true in the structure $A$ when free variables
are interpreted according to $\gamma$. For $\gamma\in \mathit{Assign}$
and $x,\many{x}$ ranging over $V$, the transitions are as follows:

\begin{minipage}[t]{0.45\linewidth}
  \begin{align*}
    &\delta(\gamma,\andsymb) = (\gamma,1)\wedge(\gamma,2) \\
    &\delta(\gamma,\orsymb) = (\gamma,1)\vee(\gamma,2) \\
    &\delta(\gamma,\allsymb{x}) = \bigwedge\nolimits_{a\in\dom(A)}(\update{\gamma}{x}{a},1)\\
    &\delta(\gamma,\existsymb{x}) = \bigvee\nolimits_{a\in\dom(A)}(\update{\gamma}{x}{a},1)\\
    &\delta(\gamma,R(\many{x})) =
      \begin{cases} \tru & \gamma(\many{x})\downarrow,\,\,A,\gamma\models R(\many{x})\\
        \fals & \text{otherwise}
      \end{cases}\\
    &  \delta(\gamma,\negsymb) = (\tilde{\gamma},1)
  \end{align*}
\end{minipage}\hfill\hspace{0.1in}
\begin{minipage}[t]{0.45\linewidth}
  \begin{align*}
    &  \delta(\tilde{\gamma},\andsymb) = (\tilde{\gamma},1)\vee(\tilde{\gamma},2)\\
    & \delta(\tilde{\gamma},\orsymb) = (\tilde{\gamma},1)\wedge(\tilde{\gamma},2)\\
    & \delta(\tilde{\gamma},\allsymb{x}) =
      \bigvee\nolimits_{a\in\dom(A)}(\tilde{\gamma'},1), \quad \gamma'=\update{\gamma}{x}{a}\\
    & \delta(\tilde{\gamma},\existsymb{x}) =
      \bigwedge\nolimits_{a\in\dom(A)}(\tilde{\gamma'},1), \quad \gamma'=\update{\gamma}{x}{a}\\
    &    \delta(\tilde{\gamma},R(\many{x})) =
      \begin{cases} \tru & \gamma(\many{x})\downarrow,\,\,A,\gamma\not\models
        R(\many{x})\\
        \fals & \text{otherwise}
      \end{cases}\\
    &  \delta(\tilde{\gamma},\negsymb) = (\gamma,1)
  \end{align*}
\end{minipage}
\vspace{0.15in}

Note that for all $(q,a)\in Q\times \alphabet'_{\FOk}$ we have
$\delta(q,a) \in \Bb^+(Q\times [\arity{a}])$, where for nullary
symbols we can take $[0]=\emptyset$.

\begin{lemma}
  $\aut_A$ accepts any sentence $\varphi$ over $\alphabet'_{\FOk}$
  that is true in $A$.
  \lemlabel{foklem}
\end{lemma}
\begin{proof}[Proof Sketch.]
  A simple induction shows that for each assignment $\gamma$
  (resp. $\tilde{\gamma}$), the language of $\aut_A$ from that state
  is precisely the set of formulas that are true (resp. false) in $A$
  under $\gamma$, i.e.,
  $L(\aut_A,\gamma) = \{\varphi(\many{x}) \,\mid\,
  \gamma(\many{x})\downarrow,\,\, A,\gamma\models \varphi(\many{x})\}$
  (resp.
  $L(\aut_A,\tilde{\gamma}) = \{\varphi(\many{x}) \,\mid\,
  \gamma(\many{x})\downarrow,\,\, A,\gamma\not\models
  \varphi(\many{x})\}$).
\end{proof}
If we fix an ordering on variables, then we can identify a state
$\gamma$ in the obvious way with a tuple of domain elements
$\many{a}$. Then the language of the automaton at $\gamma$ coincides
with the notion of \emph{logical type} for the pair
$(A,\many{a})$~\cite{libkinmodeltheory}. One consequence of this
generality is that the automaton can be seamlessly adapted to solve
the query problem for $\FOk$, as we will see in
\csecref{fokquerysepsynth}.

\subsubsection{Grammar Automaton}
\seclabel{grammar-automaton} As the name suggests, the language of an
RTG is regular, and so it is the language of some tree
automaton. Given $\grammar=\la N,\Sigma,S,P\ra$ an RTG, a
(nondeterministic) tree automaton for it is simple to define. We let
$\aut_\grammar=\la N,\Sigma,S,\delta\ra$, where for each
$B\rightarrow f(B_1,\ldots, B_{\arity{f}})\in P$ we set
$\delta(B,f) = \bigwedge_{i}(B_i,i)$. Observe that we have not used
the full power of alternation: the transition function puts at most
one condition on any given child. Thus the automaton is already
nondeterministic, which keeps the size of the final automaton
(\csecref{sepsynth-decision-procedure}) linear in the size of the
grammar. Notice also that we have made a simplifying assumption about
the form of rules in $P$, since the right-hand side could contain
subterms that are not nonterminal symbols. It is easy to show that
more complicated rules can be represented using multiple simple rules
over a larger state space, yielding an equivalent grammar of size
$\Oo(|\grammar|)$.

\subsubsection{Decision Procedure}
\seclabel{sepsynth-decision-procedure} The decision procedure for
realizability and synthesis is as follows. We define the automaton
$\aut_A$ as described above for each structure $A\in\Pos$, as well as
the automaton $\aut_\grammar$ for the grammar $\grammar$. For each
structure $B\in\Neg$ we define $\aut_B$ in the same way as for
positive structures, with one tweak: instead of initial states
$I=\{\emptyset\}$ we have $I=\{\tilde{\emptyset}\}$. We take the
product of the structure automata to get:
\vspace{-0.15in}
\begin{align*}
  \aut_\cap \,\,= \mathsmaller{\bigtimes}\limits_{M\,\in\, Pos\uplus Neg} \aut_M
\end{align*}
with number of states $\Oo(mn^k)$, where
$m=|Pos\uplus Neg|, n=\max_{M\in Pos\uplus Neg}|\dom(M)|$, and $k$ is
the number of variables. There is an exponential increase in states to
convert $\aut_\cap$ to a nondeterministic automaton $\aut'_\cap$ with
$L(\aut'_\cap)=L(\aut_\cap)$~\cite{tata,automata-logics-games}. Finally,
we take the product of $\aut'_\cap$ and $\aut_\grammar$ to get
$\aut = \aut'_\cap\times \aut_\grammar$, with
$L(\aut)= L(\aut'_\cap) \cap L(\aut_\grammar)$, and furthermore,
$L(\aut)\neq\emptyset$ if and only if there is a sentence
$\varphi\in L(\grammar)$ that separates the input structures. We solve
realizability by checking emptiness of $\aut$ in time linear in its
size, which is $\Oo(2^{\mathit{poly}(m n^k)}\,|\grammar|)$, and the emptiness
checking algorithm can construct a (small) tree if nonempty.

This construction can be easily extended to give us the following
theorem for realizability and synthesis in the full logic $\FOk$ that
includes \ite~and function terms.

\begin{theorem}
  \sepsynth{$\FOk$} is decidable in $\EXP$ for a fixed signature and
  fixed $k\in\Nat$.
  \thmlabel{fok}
\end{theorem}

The construction for the full gamut of terms is straightforward. It
can be accomplished by not only keeping partial assignments but also
states that encode the currently expected domain element for a term
under evaluation by the automaton. We give more details for how this
extension works when we discuss term synthesis in \csecref{term-synthesis}.

\subsection{Query Realizability and Synthesis}
\label{sec:fokquerysepsynth}
As noted, the automaton $\aut_A$ (\csecref{constructions}) is more
general than an acceptor of sentences, and it can be easily modified
as follows to solve query synthesis for $\FOk$ with no increase in
complexity.

For a pair $\la A,\mathit{Ans}\ra$ of a structure and an answer set,
with $\mathit{Ans}\subseteq\dom(A)^r$, we define an ATA
$\aut_A = \la Q,\alphabet,I,\delta\ra$ whose language is the set of
all formulas $\varphi(y_1,\ldots, y_{r})\in\FOk$, with $r\le k$, whose
answer set in $A$ is $\mathit{Ans}$. We describe each component
below. For simplicity, we work with a fixed permutation of $r$
distinct variables $\many{y} \in V^r$, where $V=\{x_1,\ldots,x_k\}$.

\paragraph{\textsf{States.}} The set of states is unchanged:
$Q \coloneq \mathit{Dual}(\mathit{Assign})$.

\paragraph{\textsf{Transitions.}} The transition function $\delta$ is
unchanged, with the exception of the following transitions for the
initial state $q_i=\emptyset$.

\paragraph{\textsf{Initial states.}} $I = \{q_i\}$. The main idea is
to (a) require that the automaton reads and accepts the input formula
from all states (partial variable assignments) that correspond to
tuples in the answer set $\mathit{Ans}$, and (b) reject from states
corresponding to the complement of $\mathit{Ans}$. Let
$\mathit{S(\many{y})}\subseteq Q$ be the set of assignments
defined only on $\many{y}$, and let
$\mathit{S(Ans)}\subseteq\mathit{S(\many{y})}$ be the subset
of assignments that map $\many{y}$ to a member of the answer set. For
any $a\in\alphabet'_{\FOk}$, the transition out of $q_i$ is given by:
\begin{align*}
  \delta(q_i,a) \quad &= \quad \left(\,\bigwedge_{\gamma\,\in\,
  \mathit{S(Ans)}}\delta(\gamma,a)\,\right)
  \,\,\wedge\,\, \left(\,\bigwedge_{\gamma\,\in\,
                        \mathit{S(\many{y})}\setminus \mathit{S(Ans)}}
  \delta(\tilde{\gamma}, a)\,\right)
\end{align*}

The following theorem follows easily from the proof of
\lemref{foklem}.

\begin{theorem}
  \querysepsynth{$\FOk$}~is decidable in $\EXP$ for a fixed signature
  and fixed $k\in\Nat$.
  \thmlabel{query}
\end{theorem}


\section{Realizability and Synthesis with Least Fixed Point
  Definitions}
\label{sec:recursion}

In this section we study separability for logics with least fixed
point operators. In particular, we choose a logic with a finite set of
\emph{relation variables}, each of which can be defined
recursively. These relation variables, though finite in number, can be
redefined any number of times (similar to reusing variables, as we saw
in ~\csecref{example-2}). Further, as we will discuss
in~\csecref{datalog}, the ideas presented in this section can be
extended to handle mutually recursive definitions.

Note that the ability to \emph{define} a relation is valuable
independently of recursion: with definitions we can require a formula
to be synthesized and then used in \emph{multiple distinct
  places}. For example, we may want to express: \emph{there exist x
  and y that are related in some unknown way, and further, all things
  related in that way also share a property $\psi$}. In logic, this
amounts to a separator of the form:
\begin{align*}
  \exists x. \exists y.\, \varphi(x,y)\wedge (\forall x. \forall y.\,
  \varphi(x,y)\rightarrow \psi(x,y))
\end{align*}
Notice that $\varphi$ appears twice, which we cannot express with a
regular tree grammar.  However, with relation variables and
definitions we can ask to synthesize a formula $\varphi$ in a template
as follows:
\begin{align*}
\Let\,\, R(x,y) = \varphi(x,y) \,\,\In\,\,\,
   \exists x. \exists y.\, R(x,y) \wedge (\forall x. \forall y.\,
    R(x,y) \rightarrow \psi(x,y))
\end{align*}
Following the semantics of our logic with least fixed point
definitions, we will see how the automata-theoretic approach extends
neatly to accommodate both definitions and recursion by moving from
alternating tree automata to two-way tree automata.

\subsection{First-Order Logic with Least Fixed Points}
\label{sec:folfp}
Here we describe $\FOLFP$, which is an extension of $\FO$ with
recursively-defined relations with least fixed point semantics. The
syntax is given in~\figref{FOlfp}. Formulas in $\FOLFP$ can define
relations using a set $\{P_1,P_2,\ldots\}$ of symbols disjoint from
the signature. Such symbols are interpreted as least fixed points of
the set operators induced by their definitions. Note that, although
not shown in~\figref{FOlfp}, we require all relations to be defined
before they are used.

Recall the definition for \emph{reachability} from \figref{example3}:
\begin{align*}
  \varphi \coloneq \,\,\, &\Let \,\, \mathit{reach}(x,y) =_{\lfp} (E(x,y)\vee \exists z.\,\,
  E(x,z)\wedge \mathit{reach}(z,y)) \,\,\,\In\,\,\, \varphi'(\mathit{reach})
\end{align*}
For a fixed structure $A$, the meaning of $\mathit{reach}(x,y)$ is
obtained by first interpreting the definition
$\psi(x,y,\mathit{reach}) \coloneq E(x,y)\vee \exists z.\,\,
E(x,z)\wedge \mathit{reach}(z,y)$ as a monotonic function
$F_\psi : 2^{X}\rightarrow 2^{X}$ over the lattice defined by the
subset relation on $2^X$, where $X=\dom(A)^2$. Formally, for
$Y\subseteq X$,
\begin{align*}
F_\psi(Y) \triangleq \left\{ (a_1,a_2)\in X \,\mid\,
  \psi(a_1/x,a_2/y,Y/\mathit{reach})\right\},
\end{align*}
where by $(Y/\mathit{reach})$ we mean that $\mathit{reach}$ is
interpreted as the relation $Y$ in $\psi$ (similarly for $a_1/x$ and
$a_2/y$).  Then for any $a,a'\in\dom(A)$, $\mathit{reach}(a,a')$ holds
if and only if $(a,a')\in \lfp(F_\psi)$, where $\lfp(F_\psi)$ denotes
the least fixed point of $F_\psi$.  More generally, definitions in
$\FOLFP$ are interpreted as follows in a structure $A$:
\begin{equation}
  A \models \Let\, P(\many{x}) =_\lfp \psi(\many{x},P) \, \In \,\,
  \varphi(P)
  \,\,\, \Leftrightarrow\,\,\,  A \models \varphi(\lfp(F_\psi)/P)
  \equlabel{lfpsemantics}
\end{equation}
Note that the least fixed point $\lfp(F_\psi)$ may not exist for an
arbitrary formula $\psi$. It turns out, however, that a simple
syntactic restriction can ensure existence of least fixed
points. Technically, we require all occurrences of $P$ in $\psi$ to
occur under an even number of negations. This restriction can be
enforced by the grammar, and we will not mention it further.

We consider a variant of $\FOLFP$ with a finite number of recursive
relation symbols $P$, but the symbols can be reused in later
definitions. That is, they can be shadowed. Our semantics for
definitions therefore assumes that defined relations are renamed
uniquely before applying \equref{lfpsemantics}. The semantics for the
rest of $\FOLFP$ is straightforward and follows that of $\FO$.

\begin{figure}
  \begin{minipage}[t]{0.4\linewidth}
    \begin{align*}
      \varphi\quad \Coloneqq\quad & \,\,\Let\,\, P(\many{x}) =_{\lfp} \psi \,\,\In\,\,
                \varphi\,\, \\ &\mid\,\, \exists
                x.\varphi \,\,\,\,\mid\,\, \forall x.\varphi \,\,\,\,\mid\,\,
                R(\many{t})\,\,\,\mid\,\, P(\many{t})
                \\ &\mid\,\, \varphi \vee \varphi \,\,\mid\,\,
                \varphi \wedge \varphi \,\,\mid\,\, \neg \varphi
    \end{align*}
  \end{minipage}\hspace{0.4in}%
  \begin{minipage}[t]{0.4\linewidth}
    \begin{align*}
      \psi \quad\Coloneqq\quad &\,\,\exists
             x.\psi \,\,\mid\,\, \forall x.\psi \,\,\,\,\mid\,\,\,\, R(\many{t}) \,\,\,\mid\,\,
             P(\many{t}) \\
             &\quad\quad\,\,\,\,\mid\,\, \psi \vee \psi \,\,\mid\,\,
             \psi \wedge \psi \,\,\mid\,\, \neg \psi \\
      t \quad\Coloneqq\quad &\,\, x \,\,\mid\,\, c \,\,\mid\,\, f(\many{t}) \,\,\mid\,\,
          \IfThenElse{\psi}{t}{t'}
    \end{align*}
  \end{minipage}
  \caption{Syntax for $\FOLFP$, where $\varphi$ is the starting
    nonterminal, $P$ ranges over definable relation symbols, and $R$
    and $f$ range over relations and functions from a signature. 
  }
  \figlabel{FOlfp}
\end{figure}

\subsection{Separator Realizability and Synthesis with Least Fixed Points}
\label{sec:lfpsepsynth}
We now develop a solution for $\FOLFPk$ synthesis, where $\FOLFPk$ is
the logic $\FOLFP$ restricted to $k$ first-order variables and $k'$
definable relation symbols $\{P_1,\ldots,P_{k'}\}$, which we assume
are disjoint from the relation symbols in the signature $\tau$. We
extend the ranked alphabet for $\FOk$ by representing definitions with
binary symbols $\dblqt{\Let\,P(\many{x})}$ whose children are the
definition body and the remainder of the formula that uses the
definition (superscripts denote arity):
\begin{align*}
  \alphabet_{\FOLFPk} = \left\{\, \Let\,  P(\many{x})^2,\, P^{\arity{P}}\,\,\bigm\vert\,\, P\in
  \{P_1,\ldots,P_{k'}\},\, \many{x}\in V^{\arity{P}}\,\right\}\cup  \alphabet_{\FOk}
\end{align*}

In this problem we use \emph{two-way tree automata}, which can
navigate an input tree in both directions (from a node to its children
or to its parent), thus making all parts of a tree accessible from any
node. When reading an occurrence of a definable relation symbol $P$,
the automaton can navigate to the corresponding definition, which is
elsewhere in the tree, and read it. This same capacity to move up and
down in the tree gives us an elegant way to describe the evaluation of
\emph{recursive definitions}, which must be read multiple times to
compute.

\subsubsection{Two-Way Tree Automata}
A two-way tree automaton $\aut=\la Q,\alphabet,I,\delta,F\ra$ on
$\kappa$-ary trees generalizes the transition function of ATAs to have
the form
$\delta : Q\times\alphabet \rightarrow \Bb^+(Q\times
\{-1,\ldots,\kappa\})$, where occurrences of $-1$ in a transition
require the automaton to \emph{ascend} in the input tree. Runs are
defined as for ATAs, with the exception that a run cannot ascend above
the root of the input tree. For our purposes, we modify the notion of
acceptance by distinguishing a set of final states $F\subseteq Q$. A
tree is accepted by a two-way automaton if there is a run for which
every branch reaches a state in $F$. Note that the two-way tree
automata we use here are no more expressive than alternating tree
automata, and there are algorithms to convert a two-way automaton to a
one-way nondeterministic automaton in exponential
time~\cite{two-way-vardi}, and thus emptiness is
decidable. 

\subsubsection{Automaton for Evaluating Formulas with Recursive Definitions}
\label{sec:tree-automata-foklfp}
For simplicity, we again describe a construction for the simpler
variant of $\FOLFPk$ without functions and if-then-else terms. The
ranked alphabet for this simplification looks as follows:
\begin{align*}
  \alphabet'_{\FOLFPk} = \left\{\,\Let\, P(\many{x})^2,\,
  P(\many{x})^0 \,\,\bigm\vert\,\, P\in \{P_1,\ldots,P_{k'}\}, \many{x}\in V^{\arity{P}}\,\right\} \cup  \alphabet'_{\FOk}
\end{align*}
Let us assume each symbol $P_i$ has $\arity{P_i} = r$. For a fixed
structure $A$ with $|\dom(A)|=n$, we define a two-way tree automaton
${\aut_A = \la Q,\alphabet'_{\FOLFPk},I,\delta,F\ra}$ whose language
is the set of sentences in (simplified) $\FOLFPk$ that are true in
$A$. We discuss each component next.

\paragraph{\textsf{States.}} In addition to assignments, each state
keeps track of information that enables the automaton to evaluate
recursively-defined relations. This includes (1) whether the automaton
is going \emph{up} to find a definition or \emph{down} to evaluate a
formula, (2) a counter value from
$\mathit{Count} \triangleq\{0,\ldots,n^r\}$ that tracks the stage of
the current least fixed point computation, and (3) the current
definition being evaluated (if any), i.e., a member of the set
$\mathit{Defn}\triangleq \{\bot,P_1,\ldots,P_{k'}\}$. Finally, some
states have a tuple of domain elements from
$\mathit{Val}\triangleq \dom(A)^r$ rather than a partial assignment,
which is used to pass values to the body of a definition whenever a
defined relation is used. Note that the distinction between a tuple
and an assignment takes care of part (1) above. Similar to partial
assignments, the tuples are marked to indicate the automaton's mode of
operation: checking a formula is either true (verifying) or false
(falsifying). Combining the above, we have
\begin{align*}
  Q &\coloneq
      \mathit{Dual(Assign)}\times\mathit{Count}\times\mathit{Defn}
      \,\,\cup\,\,
      \mathit{Dual(Val)}\times\mathit{Count}\times\mathit{Defn}
      \,\,\cup\,\, \{q_f\},
\end{align*}
where $q_f$ is distinct from all other states and $F=\{q_f\}$. 

The states $Q$ can be divided into two categories: \emph{up} and
\emph{down}. The \emph{up} states correspond to checking membership in
a defined relation. In an \emph{up} state
$\la \mathit{val},\mathit{count},\mathit{defn}\ra\in
\mathit{Dual}(\mathit{Val})\times\mathit{Count}\times\mathit{Defn}\subseteq
Q$, the automaton navigates up on the input tree to find the
definition for $\mathit{defn}$, and it carries a tuple $\mathit{val}$
of domain elements that it should check for membership in the defined
relation. In a \emph{down} state
$\la \mathit{assign},\mathit{count},\mathit{defn}\ra\in
\mathit{Dual}(\mathit{Assign})\times\mathit{Count}\times\mathit{Defn}\subseteq
Q$, the automaton evaluates a formula under the variable assignment
$\mathit{assign}$ while navigating down in the input tree.

\paragraph{\textsf{Initial states.}} There is one initial state
containing an empty assignment, a counter at $0$, and the current
definition set to $\bot$, i.e., $I = \{\la \emptyset,0,\bot\ra\}$.

\paragraph{\textsf{Transitions.}}
The transitions for symbols shared with $\FO$ are similar to the
earlier construction (\csecref{constructions}). The novelty is to
define transitions for definitions and occurrences of defined
relations $P_i(\many{x})$. For intuition, consider the increasing
sequence of ``approximations'' for a relation defined by a formula
$\psi$. The least fixed point for the operator $F_\psi$ (see
~\csecref{folfp}) can be computed in $n^r$ steps by iteratively
applying $F_\psi$ starting from $\emptyset$, giving us the sequence
\begin{align*}
\emptyset\subseteq F_\psi(\emptyset)\subseteq\cdots\subseteq
F_\psi^i(\emptyset)=F_\psi^{i+1}(\emptyset),
\end{align*}
where $i\le n^r$ (follows from monotonicity of $F_\psi$). When the
automaton reads a defined relation $P_i(\many{x})$ in a state
$\la \gamma, j, P_i\ra$, it will attempt to verify that
$\gamma(\many{x})=\many{a}\in F_\psi^j(\emptyset)$. Similarly, in a
state $\la \tilde{\gamma}, j, P_i\ra$ it will attempt to verify that
$\gamma(\many{x})=\many{a}\notin F_\psi^j(\emptyset)$.

Presenting all of the many transitions would obscure the main ideas,
so we give only a description of interesting ones here; a full account
can be found in~\appref{fullversion}. 
We focus on four cases: (1) reading a defined symbol, (2) finding a
definition, (3) reading a definition, and (4) a variation on (1) where
the defined symbol being read is not the current definition. Below, we
use
$\mathit{assign}\in\mathit{Dual}(\mathit{Assign}),
\gamma\in\mathit{Assign},\mathit{val}\in\mathit{Dual}(\mathit{Val}),
v\in\mathit{Val}, j\in\mathit{Count}$, and $P_i,P_j\in\mathit{Defn}$.

\begin{itemize}
\item[]
\item[(1)] \emph{\textsf{Reading a defined symbol.}} Suppose the
  automaton is reading $\dblqt{P_i(\many{x})}$ in a \emph{down} state
  $\la \mathit{assign}, j, P_i\ra$. Thus it is currently reading the
  definition for $P_i$ (call the definition $\psi$) and has
  encountered a \emph{use} of $P_i$ in the form
  $P_i(\many{x})$. Suppose $j=0$. If $\mathit{assign}=\gamma$, then
  the automaton is \emph{verifying} and it must verify that
  $v=\gamma(\many{x})$ is in the $j$th stage of the least fixed point
  computation for the definition of $P_i$. The transition for this
  case is $\fals$, since $v\notin\emptyset =
  F_\psi^0(\emptyset)$. Otherwise, if
  $\mathit{assign}=\tilde{\gamma}$, then the automaton is
  \emph{falsifying} and must verify that $v=\gamma(\many{x})$ is
  \emph{not} in the $j$th stage. So the transition for this case is
  $\tru$, since $v\notin\emptyset$.  If $j>0$, in both cases the
  transition forces the automaton to navigate up to the definition and
  evaluate it. It does this by changing to the \emph{up} state
  $\la v,j,P_i\ra$, if verifying, and to the up state
  $\la \tilde{v},j,P_i\ra$, if falsifying.
\item[]
\item[(2)] \emph{\textsf{Finding a definition.}} Suppose step (1)
  has just occurred, and thus the automaton is looking for a
  definition of $P_i$. The automaton continues moving up on the input
  tree until it encounters a symbol that marks the definition of
  $P_i$, i.e., a symbol of the form
  $\dblqt{\Let\,P_i(\many{x})}$. Note: if the definition does not
  exist in the tree, then the automaton continues to the root, at
  which point it can make no valid transition and the tree is
  rejected.
\item[]
\item[(3)] \emph{\textsf{Reading a definition.}} Suppose step (2)
  has just occurred and the automaton is in state
  $\la \mathit{val}, j, P_i\ra$ reading
  $\dblqt{\Let\,P_i(\many{x})}$. The automaton decrements the counter
  and proceeds to evaluate the definition by moving into the child
  tree corresponding to the definition of $P_i$. It enters a
  \emph{down} state with assignment $\gamma=\cup_i\{x_i\mapsto v_i\}$
  that maps variables in $\many{x}$ to the passed values in
  $\mathit{val}$. If the automaton is verifying
  ($\mathit{val}=v$), then the new state is
  $\la \gamma,j-1,P_i\ra$. Otherwise, the automaton is falsifying
  ($\mathit{val}=\tilde{v}$) and the new state is
  $\la \tilde{\gamma},j-1,P_i\ra$.
\item[]
\item[(4)] \emph{\textsf{Reading a new defined symbol.}}
  Suppose the automaton is reading $\dblqt{P_i(\many{x})}$ in a
  \emph{down} state $\la \mathit{assign}, \mathit{count}, P_j\ra$,
  where $P_j\neq P_i$. Thus it has encountered a \emph{use} of a
  defined relation $P_i$ that it is not currently reading. Rather than
  checking the value of $\mathit{count}$, as in case (1) above, the
  automaton continues on to case (2) with the current definition set
  to $P_i$ and with a fresh counter initialized at
  $\mathit{count}=
  n^r$.
\end{itemize}

\paragraph{\textsf{Acceptance.}} Recall that the automaton accepts a
tree $t$ if it has a run on $t$ where every branch reaches the final
state $q_f$. We note that the full set of transitions ensures that the
only way to reach the final state $q_f$ is via $\tru$ transitions.

\begin{lemma}
  $\aut_A$ accepts any sentence $\varphi\in\alphabet'_{\FOLFPk}$ that is
  true in $A$.
\end{lemma}

\subsubsection{Decision Procedure.}
We define automata $\aut_A$ for each input structure $A$ as described
in the construction, with the proviso that negative structures have
initial states $I=\{\la\tilde{\emptyset},0,\bot\ra\}$. We take the
product of these automata as before and convert the resulting
automaton to a one-way nondeterministic automaton without alternation
by adapting the technique of~\cite{two-way-vardi}. We further take the
product of the nondeterministic automaton with the grammar automaton
$\aut_\grammar$. Checking emptiness of the final automaton, which has
size $\Oo(2^{\mathit{poly}(m n^k k')}|\grammar|)$, gives us the
decision procedure for realizability and synthesis. Again, it is
straightforward to adapt the construction to full $\FOLFPk$ with \ite~
and function terms, giving us:
\begin{theorem}
  \sepsynth{$\FOLFPk$}~is decidable in $\EXP$ for a fixed signature
  and fixed $k,k'\in\Nat$.
  \thmlabel{lfpthm}
\end{theorem}

\section{Term Synthesis}
\label{sec:term-synthesis}

In this section we show that it is possible to use the same general
approach to \emph{synthesize terms}. We examine the term synthesis
problem for a logic similar to $\FOLFP$ with recursively-defined
functions (adaptations for other logics are similar). Much of the
construction is similar to the construction for $\FOLFP$
(\csecref{recursion}). We give the main idea by showing how the
evaluation automaton's state space changes, and we describe at a high
level some new transitions related to terms.

\subsection{First-Order Logic with Least Fixed Points and Recursive
  Functions}
\label{sec:term-synthesis-with}

We consider term synthesis for a logic with least fixed point
relations and \emph{recursively-defined functions}. Recall the list
merge example from \figref{example4}. Given a set of (unlabeled)
structures that each interpret the constants
$\mathit{in_1},\ldots,\mathit{in_d}$, and $\mathit{out}$ as an
input-output example for some functional relationship (e.g.,
$\mathit{merge(\mathit{in}_1,\mathit{in_2})} = \mathit{out}$), the
goal is to decide whether there is a term $t$ that evaluates to
$\mathit{out}$ in each structure, and to synthesize one if it
exists. We explore this problem for the language $\FOTERM$ (syntax in
\figref{forec}), which is similar to $\FOLFP$ and, additionally, has
recursively-definable functions and yields \emph{terms} rather than
formulas.

\begin{figure}
  \begin{minipage}[t]{0.4\linewidth}
    \begin{align*}
      \varphi\quad \Coloneqq\quad & \quad\,\Let\,\, g(\many{x}) =_{\lfp} t \,\,\In\,\,
           \varphi \\ &\mid\,\,\Let\,\, P(\many{x}) =_{\lfp} \psi \,\,\In\,\,
                        \varphi \\ &\mid\,\, t
    \end{align*}
  \end{minipage}\hspace{0.2in}%
  \begin{minipage}[t]{0.5\linewidth}
    \begin{align*}
      \psi \quad\Coloneqq\quad &\,\,\exists
             x.\psi \,\,\mid\,\, \forall x.\psi \,\,\,\,\mid\,\,\,\, R(\many{t}) \,\,\,\mid\,\,
             P(\many{t}) \\
             &\quad\quad\,\,\,\,\mid\,\, \psi \vee \psi \,\,\mid\,\,
             \psi \wedge \psi \,\,\mid\,\, \neg \psi \\
      t \quad\Coloneqq\quad &\,\, x \,\,\mid\,\, c \,\,\mid\,\, f(\many{t}) \,\,\mid\,\,
          \IfThenElse{\psi}{t}{t'}
    \end{align*}
  \end{minipage}
  \caption{Syntax for $\FOTERM$ with $\varphi$ the starting
    nonterminal. The language is similar to $\FOLFP$, but it yields
    terms rather than formulas and adds (recursively) definable
    functions $g$.}
  \figlabel{forec}
\end{figure}

The semantics for $\FOTERM$ coincides with $\FOLFP$ on shared
features. The only novelty is in how we define the semantics of
recursive functions and, in particular, how we account for functions
that may not be total on the domain of the structure $A$. We choose to
interpret them as \emph{partial functions} on $\dom(A)$ and to
interpret formulas in a $3$-valued logic. Note that in this setting it
is also simple and convenient to allow input structures to interpret
function symbols from the signature as partial functions. For
instance, in the $\mathit{merge}$ example from ~\figref{example4}, we
might prefer to specify $\mathit{head}$ as a partial function that is
only defined on elements that denote lists. We give here a high-level
description of the semantics for recursive functions; details can be
found in~\appref{fullversion}.

\paragraph{\textsf{Semantics for Recursive Functions.}}

A defined function $g$ with $\arity{g}=d$ is interpreted as a partial
function $g^A : \dom(A)^d\rightharpoonup \dom(A)$, which is a member
of the bottomed partial order
$\mathcal{O} =\la \dom(A)^d\rightharpoonup \dom(a), \sqsubseteq,
\bot\ra$, where $\bot$ is undefined everywhere and $f\sqsubseteq f'$
holds if for all $\many{a}\in\dom(A)^d$, whenever $f(\many{a})$ is
defined, then $f'(\many{a})$ is defined and
$f(\many{a})=f'(\many{a})$. This partial order has finite height since
all structures here are finite. Suppose $g$ is defined recursively
using a term $t(x_1,\ldots,x_d,g)$. We associate a monotone function
$F_t : \mathcal{O}\rightarrow\mathcal{O}$ to the defining term $t$,
and let $g^A$ be the least fixed point of $F_t$, which can easily be
shown to exist and to be equal to the stable point of the chain
$\bot\sqsubseteq F_t(\bot)\sqsubseteq\cdots\sqsubseteq F^i_t(\bot) =
F^{i+1}_t(\bot)$, with $i\le |\dom(A)^d|$. The syntax and semantics
for terms in $\FOTERM$ guarantee monotonicity of the function $F_t$
for any term $t$. It follows that least fixed points exist for each
definition. In general, however, care is needed to ensure that the
least fixed point is total on $\dom(A)^d$, and whether or not this is
so depends on the definition and the structure $A$.

\subsection{Automaton for Evaluating Terms with Recursive Functions}
We now sketch the main idea for defining an automaton that accepts all
terms which evaluate to a given domain element $a$ in a structure
$A$. We consider the language $\FOTERMk$, which restricts $\FOTERM$ to
$k$ first-order variables and $k'$ definable relations and functions
from $P=\{P_1,\ldots,P_{k_1}\}$ and $F=\{g_1,\ldots,g_{k_2}\}$,
respectively, with $k'=k_1+k_2$. A ranked alphabet for
$\alphabet_{\FOTERMk}$ extends $\alphabet_{\FOLFPk}$ in an obvious
way. Note that here we consider input trees representing arbitrarily
deep logical terms, in contrast to our earlier simplifications.

Fix a structure $A$ of size $n = |\dom(A)|$ and fix a domain element
$a\in\dom(A)$. We want to define a two-way automaton
$\aut_A = \la Q_{\mathsf{TERM}}, \alphabet_{\FOTERMk},I,\delta,F\ra$
that accepts the set of closed terms $t\in\FOTERMk$ such that
$t^A = a$. For simplicity, assume $r$-ary functions and relations
only. The states are as follows:
\begin{align*}
  Q_{\mathsf{TERM}}&\coloneq
                     \mathsf{EvalForm}\cup\mathsf{EvalTerm}\cup\{q_f\} \\
  \mathsf{EvalForm}&\coloneq \left(\mathit{Dual}(\mathit{Assign})\cup \mathit{Dual}(\mathit{Val})\right)\times \mathit{Count}\times
                        \mathit{Defn} \\
  \mathsf{EvalTerm}&\coloneq (\mathit{Assign} \cup \mathit{Val})\times\mathit{Count}\times\mathit{Defn}\times\dom(A)
\end{align*}
The sets $\mathit{Assign} \triangleq V\rightharpoonup\dom(A)$,
$\mathit{Val} \triangleq \dom(A)^r$,
$\mathit{Dual}(X) \triangleq \{x,\tilde{x} \,\mid\, x\in X\}$,
$\mathit{Count} \triangleq \{0,\ldots,n^r\}$, and
$\mathit{Defn} \triangleq \{\bot,P_1,\ldots,g_{k_2}\}$ serve the same
purposes as before. Notice that the automaton has a new category of
states, namely, those for evaluating terms. The transitions related to
formulas are very similar to those for $\FOLFP$. Below, we give three
representative cases for the transition function $\delta$. Let
$\gamma\in\mathit{Assign},j\in\mathit{Count}$, and
$g,\mathit{defn}\in\mathit{Defn}$.

\paragraph{\textsf{Reading variables.}}
The automaton is reading $\dblqt{x}$ and verifying that the input tree
evaluates to $a\in\dom(A)$. It only needs to check that the variable
$x$ is mapped to $a$ in the current assignment $\gamma$:
\begin{align*}
  \delta(\la\gamma,j,\mathit{defn},a\ra, x)  \,\,&=\,\,
                                                          \begin{cases}
                                                            \tru & \gamma(x)\downarrow,\,
                                                            \gamma(x)
                                                            = a \\
                                                            \fals &
                                                            \text{
                                                              otherwise
                                                            }
                                                            \end{cases}
\end{align*}

\paragraph{\textsf{Reading \ite~ terms.}} The automaton is reading
$\dblqt{\iteterm}$ and verifying that the input tree evaluates to
$a\in\dom(A)$. It must either (1) verify the condition formula (first
child) \emph{and} verify that the term in the ``then'' branch (second
child) evaluates to $a$ or (2) falsify the condition formula
\emph{and} verify that the term in the ``else'' branch (third child)
evaluates to $a$:
\begin{align*}
  \delta(\la\gamma,j,\mathit{defn},a\ra, \iteterm)  \,\,&=\,\, (\la\gamma,j,\mathit{defn}\ra,1)\wedge
                                          (\la\gamma,j,\mathit{defn},a\ra,2) \vee
                                          (\la\tilde{\gamma},j,\mathit{defn}\ra,1)\wedge
                                                          (\la\gamma,j,\mathit{defn},a\ra,3)
\end{align*}

\paragraph{\textsf{Reading a defined function $g$.}} The automaton is
reading $\dblqt{g}$ and the current definition is set to $g'$, with
$g'\neq g$, analogous to case (1) from
~\csecref{tree-automata-foklfp}. The automaton ``guesses and checks''
that the argument terms for $g$ evaluate to $\many{a}$ and ascends to
the definition of $g$ with a fresh counter set to $n^r$ to verify that
$g$ evaluates to $a$ when applied to $\many{a}$:
\begin{align*}
  \delta(\la\gamma,j,g',a\ra, g) \quad &=
                                    \bigvee_{\substack{\many{a}\,\in\, \mathit{Val}}}
                                          \left((\la\many{a},n^r,g,a\ra,-1)
                                         \wedge \bigwedge_{i\in
                                    [r]}(\la\gamma,j,g',a_i\ra,i)\right)
                                    \quad\quad\quad (g'\neq g) \\
\end{align*}

\vspace{-0.1in} The rest of the transitions follow along these lines
and are similar to those for $\FOLFP$. There is a single initial state
with $I=\{\la \emptyset,0,\bot,a\ra\}$, and the acceptance condition
is again reachability with $F=\{q_f\}$.  
Similar reasoning to that for the $\FOLFP$ construction can be used to
show:

\begin{theorem}
  \termsynth{$\FOTERMk$}~ is decidable in $\EXP$ for a fixed signature
  $\tau$ and fixed $k,k'\in\Nat$. 
\end{theorem}

The idea sketched here can be easily added to earlier constructions
without increasing complexity in order to solve synthesis for logics
with the full gamut of terms and, in particular, to solve term
synthesis for $\FOk$ in the same complexity as the separability
problem.


\section{Lower Bounds}
\seclabel{lower}

Here we present lower bounds arguing the upper bound complexity we
obtain on certain parameters of the problem is indeed tight. Given the
number of logics and variants (and problems for separators, queries,
and terms), we focus on lower bounds only for $\FOk$; of course, these
also give lower bounds for more expressive languages and variants.

\subsection{A Lower Bound for Separability in $\FOk$}

The upper bound for separability in $\FOk$ from
~\csecref{upper-bounds} is linear in the size of the grammar and
exponential in $mn^k$, where $m$ is the number of input structures and
$n$ is the maximum size of any input structure. Hence, for a fixed
$k$, the algorithm we propose is exponential time in the size of the
input. We show a matching lower bound (this can be adapted for queries
and terms as
well). 

\begin{theorem}
  \sepreal{$\FOk$}~is $\EXP$-hard for any fixed $k > 4$.
  \thmlabel{fokhardness}
\end{theorem}
\begin{proof}[Proof.]
  The reduction is from the word acceptance problem for an alternating
  polynomial space Turing machine. Given an alternating Turing machine
  $M$ and an input $w$, with $M$ using space $s$ that is polynomial in
  $|w|$, the reduction yields $s$ positively labeled first-order
  $\tau$-structures $A_1,\ldots,A_s$, with $|\dom(A_i)|=\Oo(s)$, and a
  grammar $G$ of size polynomial in $|\la M, w\ra|$. The signature
  $\tau$ depends only on $M$. Each structure consists of two parts:
  (1) a cycle of length $s$ and (2) a gadget encoding the transition
  relation for $M$ along with unique constants for tape symbols from
  the machine's tape alphabet $\Gamma$. (We use constants that encode
  the machine head and state, i.e.,
  $\Gamma'=\Gamma\times Q\cup\Gamma$.)

  Let us consider the purposes of the \emph{structures} and the
  \emph{grammar}, which are complementary. The structures can be
  viewed as distinct copies of the machine $M$ that are used to verify
  that a computation tree for $M$ on $w$, encoded in a sentence from
  the grammar, obeys the transition relation for each of the $s$ tape
  cells. The grammar can be viewed as a skeleton of \emph{computation
    trees} for $M$ on input $w$. A particular sentence $\varphi$ from
  the grammar $G$ encodes a computation tree of potentially
  exponential depth, and it asserts many things about the structure on
  which it is interpreted. When understood together across all
  structures, the truth of the assertions is equivalent to the
  computation tree being an \emph{accepting computation tree}, i.e.,
  that successive configurations follow the transition relation and
  the final configuration is accepting. The main trick is to emulate
  in the grammar the generation of each successive machine
  configuration in a way that allows checking the transition
  relation. This is not entirely straightforward because $s$ tape
  symbols cannot all be stored at once in $k$ variables ($k$ is
  fixed). Here is some intuition for how we accomplish this.


  \tikzstyle{every node}=[circle, text=black, draw, inner sep=2pt, minimum width=5pt]

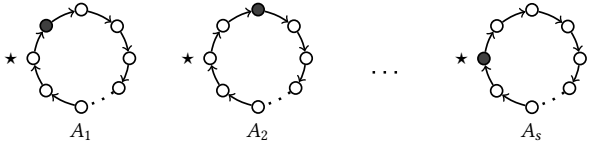
\begin{figure}[H]
  \centering
  \begin{tabular}[t]{c c c c}
  \scalebox{0.8}{
    \centering
    \begin{tikzpicture}[thick,scale=0.8]
      \node[label=left:$\star$,shape=circle,draw=black] (0) at (0,1) {} ;
      \node[shape=circle,draw=black,fill=zgray]
      (1) at (0.27,1.669) {} ; \node[shape=circle,draw=black] (2) at
      (1,2) {} ; \node[shape=circle,draw=black] (3) at (1.745,1.658)
      {} ; \node[shape=circle,draw=black] (4) at (2,1) {} ;
      \node[shape=circle,draw=black] (5) at (1.78,0.39) {} ;
      \node[shape=circle,draw=black] (6) at (1,0) {} ;
      \node[shape=circle,draw=black] (7) at (0.265,0.337) {} ;
      \node[draw=none] at (1,-0.5) {$A_1$};

      \draw[fill=black] [-to] (0) edge[black, bend left=10] (1) ;
      \draw[fill=black] [-to] (1) edge[black, bend left=10] (2) ;
      \draw[fill=black] [-to] (2) edge[black, bend left=10] (3) ;
      \draw[fill=black] [-to] (3) edge[black, bend left=10] (4) ;
      \draw[fill=black] [-to] (4) edge[black, bend left=10] (5) ;
      \draw[fill=black,loosely dotted,very thick] (5) edge[black, bend left=10] (6) ;
      \draw[fill=black] [-to] (6) edge[black, bend left=10] (7) ;
      \draw[fill=black] [-to] (7) edge[black, bend left=10] (0) ;
    \end{tikzpicture}
    \vspace{0.1in}
  } &
    \scalebox{0.8}{
    \centering
    \begin{tikzpicture}[thick,scale=0.8]
      \node[label=left:$\star$,shape=circle,draw=black] (0) at (0,1) {} ;
      \node[shape=circle,draw=black] (1) at (0.27,1.669) {} ;
      \node[shape=circle,draw=black,fill=zgray] (2) at (1,2) {} ;
      \node[shape=circle,draw=black] (3) at (1.745,1.658) {} ;
      \node[shape=circle,draw=black] (4) at (2,1) {} ;
      \node[shape=circle,draw=black] (5) at (1.78,0.39) {} ;
      \node[shape=circle,draw=black] (6) at (1,0) {} ;
      \node[shape=circle,draw=black] (7) at (0.265,0.337) {} ;
      \node[draw=none] at (1,-0.5) {$A_2$};

      \draw[fill=black] [-to] (0) edge[black, bend left=10] (1) ;
      \draw[fill=black] [-to] (1) edge[black, bend left=10] (2) ;
      \draw[fill=black] [-to] (2) edge[black, bend left=10] (3) ;
      \draw[fill=black] [-to] (3) edge[black, bend left=10] (4) ;
      \draw[fill=black] [-to] (4) edge[black, bend left=10] (5) ;
      \draw[fill=black,loosely dotted,very thick] (5) edge[black, bend left=10] (6) ;
      \draw[fill=black] [-to] (6) edge[black, bend left=10] (7) ;
      \draw[fill=black] [-to] (7) edge[black, bend left=10] (0) ;
    \end{tikzpicture}
    \vspace{0.1in}
      } &
    \scalebox{0.8}{
    \centering
    \begin{tikzpicture}[thick,scale=0.8]
      \node[draw=none] (invisible1) at (0,2) {} ;
      \node[draw=none] (invisible2) at (0,0) {} ;
      \node[draw=none] (0) at (0.2,1.5) {} ;
      \node[draw=none] (1) at (1,1.5) {} ;
      \draw[fill=black,very thick,loosely dotted] (0) edge[black] (1) ;
    \end{tikzpicture}
    \vspace{0.1in}
      }
   &
  \scalebox{0.8}{
    \centering
    \begin{tikzpicture}[thick,scale=0.8]
      \node[label=left:$\star$,shape=circle,draw=black,fill=zgray] (0) at (0,1) {} ;
      \node[shape=circle,draw=black] (1) at (0.27,1.669) {} ;
      \node[shape=circle,draw=black] (2) at (1,2) {} ;
      \node[shape=circle,draw=black] (3) at (1.745,1.658) {} ;
      \node[shape=circle,draw=black] (4) at (2,1) {} ;
      \node[shape=circle,draw=black] (5) at (1.78,0.39) {} ;
      \node[shape=circle,draw=black] (6) at (1,0) {} ;
      \node[shape=circle,draw=black] (7) at (0.265,0.337) {} ;
      \node[draw=none] at (1,-0.5) {$A_s$};

      \draw[fill=black] [-to] (0) edge[black, bend left=10] (1) ;
      \draw[fill=black] [-to] (1) edge[black, bend left=10] (2) ;
      \draw[fill=black] [-to] (2) edge[black, bend left=10] (3) ;
      \draw[fill=black] [-to] (3) edge[black, bend left=10] (4) ;
      \draw[fill=black] [-to] (4) edge[black, bend left=10] (5) ;
      \draw[fill=black,loosely dotted,very thick] (5) edge[black, bend left=10] (6) ;
      \draw[fill=black] [-to] (6) edge[black, bend left=10] (7) ;
      \draw[fill=black] [-to] (7) edge[black, bend left=10] (0) ;
    \end{tikzpicture}
    \vspace{0.1in}
      }
  \end{tabular}
  \vspace{-0.1in}
  \caption{Structure $A_i$ tracks the window centered on the $i$th
    cell and has a special cycle element (the dark node) of distance
    $i$ from the start of the cycle, denoted $\star$. As the symbols
    of a configuration are produced, the variables $\many{y}$ are
    equated with the current tape window if the cycle pointer is equal
    to the dark element. }
  \figlabel{reduction-figure}
\end{figure}

  Each of the $s$ structures is made to track a distinct window of
  three contiguous tape cells (as well as the previous contents for
  the window). The grammar uses a polynomial-sized gadget of
  nonterminals to iteratively produce the tape cell contents of a
  given configuration. Refer to~\figref{reduction-figure} in the
  following for a picture of how this gadget works over each
  structure. In each iteration, the grammar moves a $\mathit{ptr}$
  variable along the cycle of size $s$. If $\mathit{ptr}$ is equal to
  a special element, filled dark in~\figref{reduction-figure}, then
  the grammar requires the current tape cell's contents (and
  neighbors) to be stored in variables by asserting an equality. Each
  structure is made to track a unique window by differently
  interpreting the distance between a starting node, denoted $\star$,
  and the special dark node. The grammar $G$ ensures the following
  invariant holds for all sentences $\varphi\in L(G)$: if we evaluate
  $\varphi$ in $A_i$, then upon evaluating the subformula of $\varphi$
  that picks a symbol for cell $i+1$, window $i$ of the previous
  configuration is stored in variables $x_1,x_2,x_3$ and window $i$ of
  the current configuration is stored in variables $y_1,y_2,y_3$. The
  grammar checks that successive windows obey the transition relation
  by asserting the relation
  $\delta(\mathit{choice},\many{x},\many{y})$, where $\delta$ encodes
  the transition relation for $M$ and $\mathit{choice}\in\{0,1\}$
  encodes which of two transitions for the alternating machine is
  being verified.

  More details can be found in~\appref{fullversion}.
\end{proof}

\begin{theorem}
  \querysepreal{$\FOk$}~is $\EXP$-hard for fixed $k > 4$.
\end{theorem}
\begin{proof}[Proof.]
  Reduction from \sepreal{$\FOk$}. Positive structures have a full
  query answer set and negative structures have an empty query answer
  set.
\end{proof}

\begin{theorem}
  \sepreal{$\FOLFPk$}~is $\EXP$-hard for fixed $k,k'\in\Nat$ with
  $k > 4$.
\end{theorem}
\begin{proof}[Proof.]
  Reduction from \sepreal{$\FOk$}.
\end{proof}

\subsection{More Lower Bounds and Open Problems}

We can now ask whether separator realizability for $\FOk$ is decidable
in polynomial time if there is only \emph{one} structure. With only
one structure (positively labeled, say) the problem of separability
may seem odd, but checking whether there is \emph{any} sentence in the
grammar $G$ that is true on the single structure is actually a
nontrivial problem; indeed, the grammar is quite powerful.

We can show a general reduction from separator realizability for
multiple structures to realizability for a single structure (but over
a different grammar and for $k'=k+1$). Given a set of positive
structures $\mathit{Pos}$, negative structures $\mathit{Neg}$, and a
grammar $G$ (over $\alphabet_{\FOk}$), we can reduce the realizability
problem to a new realizability problem over a single positive
structure $M$ and a grammar $G'$ (over $\alphabet_{\FO(k+1)}$). The
idea is that $M$ has (a) copies of all the structures in
$\mathit{Pos}$ and $\mathit{Neg}$, (b) a set of elements $i$, one for
each $i \in |\mathit{Pos} \uplus \mathit{Neg}|$, which represent
\emph{structure identifiers}, (c) unary relations $\mathit{Id}$ and
$\mathit{P}$, where $\mathit{Id}$ holds for the set of structure
identifiers $i$ and $\mathit{P}$ holds for the set of identifiers for
structures in $\mathit{Pos}$, and (d) a binary relation
$\mathit{Owns}$ that associates each $i$ with the elements of the copy
of structure $i$ in $M$. The grammar $G'$ is designed to generate only
formulas of the form
$\forall i. \mathit{Id}(i) \rightarrow (\mathit{P}(i) \leftrightarrow
\alpha'(i))$. The formula $\alpha'$ is obtained by taking a formula
$\alpha$ admitted by $G$ and relativizing the quantification so that
it is restricted to those elements that are associated to $i$. For
example, $\forall x. \beta(x)$ is relativized to
$\forall x. \mathit{Owns}(i,x) \rightarrow \beta(x)$ and
$\exists x. \beta(x)$ is relativized to
$\exists x. \mathit{Owns}(i,x) \wedge \beta(x)$. A formula
$\varphi'\in L(G')$ is true in $M$ if and only if there is a formula
$\varphi\in L(G)$ that is a separator for $\mathit{Pos}$ and
$\mathit{Neg}$. Note that the formulas in $G'$ use one extra variable
(namely $i$).

This reduction combined with~\thmref{fokhardness} shows the following:
\begin{theorem}
  For any fixed $k>5$, given a single structure $M$ and a RTG
  $\,\grammar$ over $\FOk$, checking whether there is a formula in
  $L(G)$ that is true in $M$ is $\EXP$-complete.
\end{theorem}

\paragraph{\textsf{Open Problems.}} Our algorithm has exponential
dependence on the number of structures $m$. We do not know whether
algorithms polynomial in $m$ can be achieved. More precisely, we do
not know if separator realizability can be achieved in time
$\Oo(f(m,n,k)\cdot g(n,k))$, where $f$ is a polynomial function and
$g$ is an arbitrary function. Learning algorithms that scale linearly
or polynomially with the number of data samples are clearly desirable.

Interestingly, if there is no grammar restriction, i.e., we look for a
separator in $\FOk$, then such an algorithm is indeed possible. This
follows from a suggestion by Victor
Vianu~\cite{vianupersonalcommunication}. The algorithm works on the
basis of $\FOk$-types~\cite{libkinmodeltheory}, which capture
equivalence classes of finite structures that cannot be distinguished
from each other by any $\FOk$ formula. Of crucial importance is the
fact that these equivalence classes of structures can be
\emph{defined} by an $\FOk$ formula, which can be effectively computed
for a given structure. Consequently, we can \emph{independently}
compute the defining formula, denoted $\mathit{type}(A_i)$, for each
$A_i\in\mathit{Pos}$ and then form the disjunction
$\psi\coloneq \bigvee_i \mathit{type}(A_i)$. If $\psi$ holds for any
structure in $\mathit{Neg}$ then there can be no separator. Otherwise
$\psi$ is a separator. This procedure works in time polynomial in the
number of structures.

However, we do not see any way to adapt the above procedure to
arbitrary grammars. Furthermore, it has a disadvantage as an algorithm
for learning--- it yields very large formulas that essentially overfit
the positive samples. In contrast, the automata-theoretic method can
find the smallest formulas.

There are, of course, many lower bound problems that are open for
different logics and variants, and each of them has many parameters
($|G|$, $m$, $n$, $k$, $k'$, $k_1$, $k_2$, as well as the arities of
symbols). One can ask several parameterized
complexity~\cite{parameterized-complexity-flum-grohe} lower bound
questions for each of our problems, and we leave this to future
work. In particular, one key question involves the parameter $k$
(which we have assumed is fixed in most of our treatment): is the
double exponential dependence on $k$ tight?


\section{Further Results and Discussion}
\label{sec:further-results}

In this section we discuss how the technique illustrated in
~\csecref{upper-bounds},~\csecref{recursion},
and~\csecref{term-synthesis}, can be adapted to solve problems in two
other settings: logic programming and second- and higher-order
logics. We also remark on the generality of the approach and give a
connection to Ehrenfeucht-Fra{\"i}ss{\'e} games.

\subsection{Mutual Recursion and Logic Programming}
\label{sec:datalog}

Recall that our treatment of $\FOLFP$ from \csecref{recursion} did not
include \emph{mutually-recursive} definitions. In fact, mutual
recursion can be handled with a modest increase in the number of
automaton states. Consider a variant of $\FOLFP$ that allows
\emph{blocks} of defined relations, in which all relations in a single
block can refer to each other in their definitions, like the
following:
\begin{align*}
  \Let\,\,
  \Bigg\{ \,\begin{aligned}
    P_1(x_1,x_2) \, =_{\lfp}\, \varphi_1(x_1,x_2,P_1,P_2) \\
    P_2(x_1,x_2) \, =_{\lfp}\, \varphi_2(x_1,x_2,P_1,P_2) \\
  \end{aligned} \,\Bigg\} \,\,
  \In\,\, \varphi(P_1,P_2)
\end{align*}

The semantics for blocks of mutually-recursive definitions can be
defined in terms of a \emph{simultaneous fixed point}. For the example
above, we can define functions
$F_{1}, F_{2} : 2^X\times 2^X\rightarrow 2^X$, where $X=\dom(A)^2$,
and for $X_1,X_2\subseteq 2^X$:
\begin{equation*}
\begin{split}
  F_{i}(X_1,X_2) \triangleq \left\{\, \many{a}\in X \,\bigm\vert\,
    A\models \varphi_{i}(\many{a}/\many{x},X_1/P_1,X_2/P_2) \,\right\}
  \quad i\in\{1,2\}
  \label{simult}
\end{split}
\end{equation*}
We can interpret the relations $P_1$ and $P_2$ as the components of
the simultaneous least fixed point of the system of equations above;
see~\cite{Fritz2002} for more on simultaneous fixed points.

\paragraph{\textsf{Evaluating Mutually-Recursive Definitions.}}
In the spirit of our technique, we ask how an automaton can check
membership for a relation defined by mutual recursion using state
bounded by the structure. The same ideas carry over from
\csecref{recursion} with a modification. As before, all tuples can be
associated with the stage at which they enter the (now)
\emph{simultaneous} fixed point computation, and the automaton can use
counters to check membership at a given stage. However, the number of
stages grows exponentially in the number of relations in a block of
definitions (which we can assume is bounded by the number of definable
symbols $k'$). The automaton state must now include a \emph{product}
of counters, one for each definable symbol. Other than this change to
the states, the construction that handles mutual recursion in $\FOLFP$
remains essentially the same.

\paragraph{\textsf{Logic Programming.}}
\label{sec:datalog-synthesis}
With mutually-recursive definitions, our technique can be used to
solve $\Datalog$ synthesis problems; this is not surprising since
$\Datalog$ is logically similar to standard first-order logics with
least fixed points (in fact, it corresponds to an existential fragment
$\exists \LFP$ of first-order logic with least fixed points that only
allows negation on atomic relations from the signature and disallows
universal
quantification~\cite{libkinmodeltheory}). See~\appref{fullversion} for
more details on a $\Datalog$ synthesis problem that our technique can
solve. (We note for a fixed number of variables and definable
relations, the space of $\Datalog$ programs is finite and thus
decidability is not theoretically interesting.) We can also model
problems from \emph{inductive logic programming} (ILP)~\cite{ilp},
e.g., learning from entailment over bounded variable horn-clause
programs, by encoding background knowledge (a set of definite horn
clauses) in the grammar.

\subsection{Second-Order Logic}
\label{sec:second-order-logics}
The approach naturally extends to second-order logic ($\SO$) (see,
e.g.,~\cite{libkinmodeltheory} for syntax and semantics). We state
here a result for relational $\SOk$, a version of $\SO$ restricted to
$k$ first-order variables and $k'$ second-order relation variables. An
alternating one-way automaton $\aut_A$ can evaluate $\SO(k,k')$
formulas on a fixed structure $A$ by keeping track of an assignment to
$k$ first-order variables and an assignment to $k'$ second-order
relation variables of maximum arity $r$ (the relation variables map to
sets of $r$-tuples) using a state space of size
$\Oo(2^{\left(k' n^r\right)} n^{k})$. The decision procedure follows
the same lines as before.

\begin{theorem}
  \sepsynth{$\SOk$}~is decidable in $2\EXP$ for a fixed signature and
  fixed $k,k'\in\Nat$. 
\end{theorem}

The same idea sketched above easily extends to logics with variables
over higher-order functions.

\subsection{Discussion}
We believe the tree automata-theoretic approach proposed in this work
is extremely versatile. The crux is to build automata that, when
reading the parse tree of an expression, can evaluate it on a fixed
structure using finitely many states. This typically is true if there
is a way to recursively evaluate the semantics of expressions using
memory that depends on the size of the structure but \emph{not on the
  size of expressions}. Bounding the number of variables is one way to
achieve this.

We claim our technique applies to any logic or language for which (a)
the semantics of expressions can be described locally in the parse
tree in terms of the semantics for subexpressions and (b) evaluating
the semantics at each node of the parse tree requires memory that is
bounded by a function of the structure size (and not the formula
size). The fact that logics with definitions and recursion can be
captured with two-way tree automata shows that they also meet these
conditions, since the automaton can be converted to a deterministic
bottom-up automaton. We leave formalizing this claim, proving it, and
finding further instantiations of the technique to future work.

Finally, we note that the separability problems considered here can be
viewed from the perspective of Ehrenfeucht-Fra{\"i}ss{\'e}
games~\cite{Ehrenfeucht1961, fraisse_1953}, which are typically used
to show formulas in a logic can or cannot distinguish between two
structures. The separability problem instead asks whether a set of
positive structures can be separated from a set of negative structures
using formulas that conform to a given grammar. Hence the game in our
setting is one that is specific to the given grammar and furthermore
forces the players to play \emph{simultaneously} on all the
structures. We leave further investigation of this relationship to
future work.


\section{Related Work}
\seclabel{related}

\emph{Program Synthesis from Examples.} Learning logical formulas is
closely related to \emph{program synthesis}, and especially, program
synthesis from \emph{examples} (as opposed to deductive approaches
from specifications~\cite{manna-waldinger}). Synthesis from examples,
or \emph{programming by examples} (PBE), has been active in recent
years and has seen successes in practice (e.g., ~\cite{flashmeta}). In
PBE, the goal is to synthesize a program consistent with a set of
input-output examples; several domains have been explored, e.g.,
synthesis of database queries from
examples~\cite{discovering-queries-sigmod14,synthesizing-sql-cheung17,thakkar-example-guided-synth}
and from analysis of database-backed application
code~\cite{optimizing-database-apps-lezama13}, synthesis of data
completion scripts~\cite{fta-data-completion-scripts}, data structure
transformations~\cite{data-structure-transformations}, and typed
functional
programs~\cite{type-directed-synth-osera,synth-refinement-types-nadia16}. A
common approach to PBE involves \emph{version space
  algebra}~\cite{tom-mitchell-vsa}, where the idea is to capture the
set of all programs that work on each example in a compact
representation and then intersect the sets for each example to
represent programs consistent with all examples (e.g.,
see~\cite{flashfill}). Our approach essentially uses tree automata as
a version space algebra to capture \emph{all} logical expressions that
satisfy some criterion over input structures.

\emph{Synthesis with Grammar.} Using grammar to constrain the
hypothesis space follows a line of work in program synthesis that uses
syntactic biases like partial programs, e.g. ~\cite{sketch}, and more
broadly, \emph{syntax-guided synthesis} (SyGuS) for logics (typically
logics supported by SMT theories)~\cite{sygusJournal}. In a SyGuS
problem, one is given a grammar from which to synthesize a logical
expression (similar to the setting in this paper) as well as a
specification in the form of a universally-quantified formula that
refers to a placeholder $e$, which must be valid when the synthesized
expression is plugged in for $e$. The separability problem can in fact
be formulated as a SyGuS problem, though SyGuS divisions and tools
only support synthesis of \emph{quantifier-free formulas}. There is a
large body of work exploring program synthesis and syntax-guided
synthesis for quantifier-free logics that focuses on practical and
scalable techniques, and for the most part does not offer any
guarantee of completeness. When grammars admit infinitely many
expressions, these solvers cannot report \emph{unrealizability}, which
is in general undecidable~\cite{caulfieldarxiv}. The ability to decide
realizability and synthesis is a crucial difference in our work.

\emph{Decidable Realizability and Synthesis.} For systems and programs
that have \emph{finite state spaces}, the realizability problem has
been extremely well studied and a rich class of specifications for
such systems is known to admit decidable realizability. The crux of
the techniques used in this domain rely on tree automata that work on
\emph{infinite} trees and infinite games played on finite graphs
(while our work uses tree automata on finite trees).  This problem was
first proposed by Church~\cite{church60}, and a rich theory of
realizability/synthesis has
emerged~\cite{BuchiLandweber69,Rabin72,automata-logics-games,kpvPneuli,PR89,KMTV00,PR90,madhudistsynth}.
The key idea is to encode the branching behavior of a reactive system
using an infinite tree and build automata that accept systems (trees)
whose behaviors satisfy a specification.

Our work is technically closer to the approach in~\cite{madhuCSL11},
which studies synthesizing \emph{imperative reactive programs} over a
finite number of variables ranging over finite domains with logical
specifications (e.g., linear temporal logic). The decidability of
realizability/synthesis is proved using tree automata that work on
finite trees (parse trees of programs), similar to the work presented
here. Unlike our work, the tree automata have infinitary acceptance
conditions in order to capture properties of infinite executions of
programs. Other differences include (1) our work interprets logical
expressions over \emph{unbounded} structures, and (2) the
specification for synthesis is not a logical formula, but rather a set
of labeled structures. Intuitively, we trade the power of logical
specifications in~\cite{madhuCSL11} and replace it with a finite set
of structures in order to synthesize over unbounded domains. Though
the constructions in our work are too large to implement na{\"i}vely,
the core idea to use tree automata on parse trees of expressions for
synthesis has been made practical in some recent work, e.g., for
string and matrix transformations ~\cite{Wang2017} and string
encoders/decoders and comparators ~\cite{WangWangDilligOOPSLA18}.

Decidability results for synthesis of expressions over unbounded data
domains are uncommon, though there are some recent results for
restricted classes of programs and models of computation, e.g.,
synthesizing finite-state transducers~\cite{khalimov18} and
synthesizing a restricted class of imperative
programs~\cite{uninterpretedsynth}. In~\cite{uninterpretedsynth}, the
authors study the problem of synthesizing uninterpreted imperative
programs from a given grammar, where programs come with assertions
that must be satisfied for any interpretation of function and relation
symbols over any domain, possibly infinite. For the restricted
subclass of coherent programs~\cite{coherence2019}, there is a
decision procedure based on tree automaton emptiness, and, similar to
our work, the solution uses tree automata working over parse
trees. There is also recent work giving sound techniques for proving
unrealizability of SyGuS problems~\cite{HuUnrealizability2019}, and,
more recently, a decision procedure for SyGuS problems over linear
integer arithmetic with conditionals over finitely-many examples
~\cite{reps20-unrealizability}.

\emph{Learning Logical Formulas.} In~\cite{aiken-fo-sep}, the authors
study a separability problem for first-order logic formulas with
\emph{bounded quantifier depth}. In contrast to the problems we
consider in this work, bounding the quantifier depth makes the search
space finite up to logical equivalence, enabling a reduction to and
from SAT. There is also work on the decidability of learning
separators from labeled examples for various description
logics~\cite{lutz-ijcai2019,lutz-2020}.  There, separation problems
are studied in the presence of an ontology, which is a finite set of
logical sentences. The presence of ontologies makes the problem
different from our work; adapting our general synthesis approach to
the world of description logics remains future work. There is also
prior work studying the complexity of learning logical concepts by
characterizing the VC-dimension of logical hypothesis
classes~\cite{grohe-logic-learning-model}, work on parameterized
complexity for logical separation problems in the PAC
model~\cite{van-bergerem-2021}, learning $\MSO$-definable concepts on
strings~\cite{learning-mso-on-strings} and concepts definable in
first-order logic with counting~\cite{learning-in-fo-with-counting},
learning temporal logic formulas from
examples~\cite{neider-ltl-learning}, and learning quantified
invariants for arrays~\cite{madhu-qda}.

\emph{Inductive Logic Programming.} In \emph{inductive logic
  programming} (ILP)~\cite{ilp}, the goal is to learn a logic program
from data, typically positive and negative examples of a target
relation. ILP systems can learn from a small number of examples and
with background knowledge (e.g., a set of horn clauses), and some
systems are able to invent new predicates and learn programs with
recursion~\cite{ilp-turning30}. Typically, ILP systems learn Prolog
programs, but recent work has explored learning in restricted
hypothesis spaces for logic programs, e.g.,
$\Datalog$~\cite{aws-synth-datalog,evans-greffen-noisy} and answer set
programming~\cite{ilasp}. As discussed in~\csecref{datalog}, our
approach can be used to model some forms of ILP by encoding background
knowledge in the grammar, and it seems possible that aspects of
\emph{metarules}~\cite{mil-muggleton} can also be achieved with our
technique; exploring connections to ILP is an interesting avenue for
future work.


\section{Conclusion}
\seclabel{conclusion}

We have argued for a very general tree automata-theoretic approach to
learning logical formulas and, more generally, any expression which
can be evaluated using state dependent on a background structure but
independent of the expression size. This is the case for the finite
variable logics studied in this work, as well as higher-order logics
and logics with fixed point operators over finite
structures. Precisely characterizing the power of this approach is an
interesting direction for future work, and so too are the lower bounds
and parameterized complexity questions we leave open.

What is nice about the tree automaton-based approach advocated here is
that various infinite concept spaces constrained by a grammar can be
seen to have finitely-many equivalence classes \emph{modulo example
  structures}. Indeed, the states of the (minimal) automaton
correspond to equivalence classes of formulas from the grammar that
are equivalent with respect to the given input structures. Effective
emptiness checking algorithms for tree automata show that we only need
to keep a single representative from each equivalence class to solve
synthesis. 
Exploring practical algorithms for restricted grammars and classes of
structures, including learning in the presence of background theories
(such as arithmetic, used say for counting), are intriguing directions
for future work.

\begin{acks}                            
  We thank Victor Vianu for discussions and for suggesting the connection to logical
  types. This work was supported in part by a Discovery Partner’s
  Institute (DPI) science team seed grant and a research grant from
  Amazon.
\end{acks}

\bibliography{main}

\ifappendix
\newpage
\appendix
\section{Two-way tree automata}
\applabel{two-way-details}

\subsubsection{Two-way tree automata}
\label{sec:two-way-tree}
Two-way (alternating) tree automata over $\kappa$-ary trees generalize
ATAs by enabling them to transition not only down into children but
also \emph{up} to the parent on the input tree. This extra capability
is formally captured by generalizing the transition function from the
form
$\delta : Q\times\Sigma\rightarrow \Bb^+(Q\times\{1,\ldots,\kappa\})$
to the form
$\delta : Q\times\Sigma \rightarrow
\Bb^+(Q\times\{-1,\ldots,\kappa\})$, where $-1$ means \emph{ascending}
to a parent node and $0$ means staying at the current node. In the
context of two-way automata only, we use $[\kappa]$ to denote the set
$\{-1,\ldots,\kappa\}$ for $\kappa\in\Nat$. For each
$(q,a)\in Q\times\alphabet$ we require
$\delta(q,a)\in \Bb^+(Q\times[\arity{a}])$. For
$x\in\Nat^*, j\in\Nat,$ let $(x\cdot j)\cdot -1 = x$ and
$(x\cdot j)\cdot 0 = x\cdot j$ and let $\epsilon\cdot -1$ be
undefined.

A two-way tree automaton is a tuple $\aut=\la Q,\Sigma,I,\delta,F\ra$,
where the only difference with ATAs is the acceptance condition
$F\subseteq Q$ and the definition of $\delta$ as described above. The
notion of a run in the two-way case is identical to the definition for
one-way alternating automata from \csecref{ata}, modulo the difference
in $\delta$. Nevertheless, we repeat the definition here and encourage
the reader to observe that the nodes of a run may be labeled by
positions in the input tree that go both up and down. A \emph{run} for
a two-way tree automaton $\aut=\la Q,\Sigma,I,\delta,F\ra$ on a tree
$t\in T_\Sigma$ is an ordered tree $\rho$ over $Q\times \treepos(t)$
satisfying the following two conditions:
\begin{itemize}
\item $\rho(\epsilon) = (q_i,\epsilon)$ for some state $q_i\in I$
\item Let $p\in \treepos(\rho)$. If $\rho(p) = (q,x)$ with
  $t(x) = a$, then there is a subset
  $S = \{(q_{1},i_1),\ldots,(q_{l},i_l)\}\subseteq Q\times
  [\arity{a}]$ such that $S\models \delta(q,a)$ and
  $\rho(p\cdot j) = (q_j,x\cdot i_j)$ for $1\le j\le l$.
\end{itemize}
Observe that a run for a two-way automaton may be infinite on a finite
input tree. We want only finite runs to be accepting, rather than
simply requiring the existence of a run, as we did for ATAs. This
corresponds to a \emph{reachability} acceptance condition, wherein a
run is \emph{accepting} if every branch reaches some state $q_f\in F$,
and a two-way automaton \emph{accepts} a tree if it has an accepting
run on it.

Note that the two-way tree automata we use here are no more expressive
than alternating tree automata, and there are algorithms to convert a
two-way automaton to a one-way automaton~\cite{two-way-vardi}, and
thus membership and emptiness are decidable. We next sketch this
conversion.

\section{Two-way tree automata to one-way ($\csecref{tree-automata-foklfp}$)}
\applabel{two-way-conversion-gist}

We can convert the two-way automaton into a language-equivalent
nondeterministic automaton by adapting the technique
of~\cite{two-way-vardi}. The key idea is to view the membership
problem of the two-way automaton as a finite reachability game, where
a Protagonist is trying to show the automaton has an accepting run and
the Antagonist is trying to refute this. Game positions are pairs of
automaton states and positions on the input tree. From the position
$(q,x)\in Q\times\treepos(t)$, with $t(x) = a\in\alphabet$, the
Protagonist picks a move $S\subseteq Q\times [\arity{a}]$ and the
Antagonist responds by picking $s\in S$, with the game continuing from
the position indicated by $s$. Play begins from the state
$(q_i,\epsilon)$. The Protagonist wins if she has a winning strategy
to reach game positions of the form $(q_f,x)$ for $x\in\treepos(t)$,
otherwise the Antagonist wins. These notions are standard for finite
reachability games, and we refer the reader to ~\cite{Fritz2002} for
details.

We can define a top-down nondeterministic automaton (without
alternation) that reads trees annotated with winning strategy
information. Since the game is a finite reachability game, it is
determined with memoryless strategies
(see~\cite{automata-logics-games}). That is, each position in the game
is won by a single player using a strategy that does not depend on the
preceding history of plays. It follows that the strategy annotations
can be represented with a finite alphabet. There are now three related
challenges: (1) the strategy must be verified to comply with the
transition function, (2) all plays that could result from the strategy
must be winning, and (3) these must both be accomplished in a single
downward pass over the annotated input tree. The final
nondeterministic automaton over (unannotated) logical formulas is
obtained by projecting out the annotation, appealing to closure of
tree regular languages under homomorphisms.

The strategy annotation for a given node $x\in\treepos(t)$ in an input
tree $t$ must decide which states to go into and in what directions on
$t$ to go from each possible game position at $x$. The strategy
annotation associates to each node $x$ of $t$ a member of
$2^{Q\times[\arity{t(x)}]\times Q}$. At a game position $(q,x)$, the
strategy for the Protagonist corresponds to a subset
$X = \{ (d,q') \,\mid\, (q,d,q') \} \subseteq
2^{Q\times[\arity{t(x)}]\times Q}$, elements of which indicate a
direction on the input tree and a next state. Thus the final
nondeterministic automaton has number of states exponential in the
number of states for the two-way automaton, and this conversion can be
done in time exponential in the size of the two-way automaton.

\section{$\FOLFP$ automaton transitions ($\csecref{tree-automata-foklfp}$)}
\applabel{lfp-transitions}

Recall we have fixed a structure $A$ and we are defining the
transitions for a two-way automaton $\aut_A$ that accepts the set of
$\FOLFPk$ formulas that are true in $A$. Recall the state space for
this automaton is defined as:
\begin{align*}
  Q &\coloneq
      \mathit{Dual}(\mathit{Assign})\times\mathit{Count}\times\mathit{Defn}
      \,\,\cup\,\,
      \mathit{Dual}(\mathit{Val})\times\mathit{Count}\times\mathit{Defn}
      \,\,\cup\,\, \{q_f\},
\end{align*}

In the rules below, $\gamma$ ranges over
$\mathit{Assign}\coloneq V\rightharpoonup \dom(A)$, $d$ ranges over
$\mathit{Dual}(\mathit{Val})$, $a$ ranges over $\mathit{Val}$, the
counter value $j$ ranges over
$\mathit{Count}\coloneq \{0,\ldots,n^r\}$, and $Y$ ranges over
$\mathit{Defn}\coloneq \{\bot, P_1,\ldots, P_{k'}\}$. All transitions
not covered below are $\fals$.

\begin{itemize}
\item[--] $\delta((\gamma, j, Y),\andsymb) = ((\gamma,j,Y),1)\wedge((\gamma,j,Y),2)$
\item[--] $\delta((\gamma, j, Y),\orsymb) = ((\gamma, j,Y),1)\vee((\gamma, j,Y),2)$
\item[--] $\delta((\gamma, j, Y),\allsymb{x}) =
  \bigwedge_{a\in\dom(A)}((\update{\gamma}{x}{a}, j, Y),1)$
\item[--] $\delta((\gamma, j, Y),\existsymb{x}) =
  \bigvee_{a\in\dom(A)}((\update{\gamma}{x}{a}, j, Y),1)$
\item[--]
  $\delta((\gamma, j, Y),R(\many{x})) = \begin{cases} \tru &
    \gamma(\many{x})\downarrow \,\,\text{and}\,\,A,\gamma\models
    R(\many{x})\\
    \fals & \text{otherwise}\\
  \end{cases}$
\item[--] $\delta((\gamma, j, Y),\negsymb) = ((\tilde{\gamma}, j, Y),1)$
\item[--] $\delta((\tilde{\gamma}, j, Y),\andsymb) = ((\tilde{\gamma},
  j, Y),1)\vee((\tilde{\gamma}, j, Y),2)$
\item[--] $\delta((\tilde{\gamma}, j, Y),\orsymb) = ((\tilde{\gamma},
  j, Y),1)\wedge((\tilde{\gamma}, j, Y),2)$
\item[--] $\delta((\tilde{\gamma}, j, Y),\allsymb{x}) =
  \bigvee_{a\in\dom(A)}((\tilde{\gamma'}, j, Y),1)$ \quad where $\gamma'=\update{\gamma}{x}{a}$
\item[--] $\delta((\tilde{\gamma}, j, Y),\existsymb{x}) =
  \bigwedge_{a\in\dom(A)}((\tilde{\gamma'}, j, Y),1)$ \quad where $\gamma'=\update{\gamma}{x}{a}$
\item[--] $\delta((\tilde{\gamma}, j, Y),R(\many{x})) = \begin{cases} \tru &
    \gamma(\many{x})\downarrow \,\,\text{and}\,\,A,\gamma\not\models
    R(\many{x})\\
    \fals & \text{otherwise}\\
  \end{cases}$
\item[--] $\delta((\tilde{\gamma}, j, Y),\negsymb) = ((\gamma, j, Y),1)$
\item[--] $\delta((d, j, Y), \andsymb) = ((d,j,Y),-1)$
\item[--] $\delta((d, j, Y), \orsymb) = ((d,j,Y),-1)$
\item[--] $\delta((d, j, Y), \allsymb{x}) = ((d,j,Y),-1)$
\item[--] $\delta((d, j, Y), \existsymb{x}) = ((d,j,Y),-1)$
\item[--] $\delta((d, j, Y), \negsymb) = ((d,j,Y),-1)$
\item[--] $\delta((\gamma, j, Y),\Let\,\,P_i(\many{x})) =
  ((\gamma, j, Y), 2)$
\item[--] $\delta((\gamma, j, Y), P_i(\many{x})) =
  ((a, n^r, P_i), -1)$ for $Y\neq P_i$,
  $\gamma(\many{x})\downarrow$, and $a=\gamma(\many{x})$
\item[--]
  $\delta((\tilde{\gamma}, j, Y), P_i(\many{x})) = ((\tilde{a}, n^r,
  P_i), -1)$ for $Y\neq P_i$, $\gamma(\many{x})\downarrow$, and
  $a=\gamma(\many{x})$
\item[--] $\delta((d, j, Y),\Let\,\,P_i(\many{x})) = ((d, j, Y), -1)$
  for $Y\neq P_i$
\item[--] $\delta((a, j, P_i),\Let\,\,P_i(\many{x})) =
  ((\gamma, j-1, P_i), 1)$ for $j>0$ and $\gamma=\cup_i\{x_i\mapsto a_i\}$
\item[--] $\delta((\tilde{\gamma}, 0, P_i),P_i(\many{x})) = \tru$
\item[--] $\delta((\tilde{a}, j, P_i),\Let\,\,P_i(\many{x})) =
  ((\tilde{\gamma}, j-1, P_i), 1)$ for $j>0$ and $\gamma=\cup_i\{x_i\mapsto a_i\}$
\item[--] $\delta((\gamma, 0, P_i),P_i(\many{x})) = \fals$
\item[--] $\delta((\gamma, j, P_i), P_i(\many{x})) =
  ((a, j, P_i), -1)$ for $j>0$ and $\gamma(\many{x})\downarrow$ and $a=\gamma(\many{x})$
\item[--] $\delta((\tilde{\gamma}, j, P_i), P_i(\many{x})) =
  ((\tilde{a}, j, P_i), -1)$ for $j>0$ and $\gamma(\many{x})\downarrow$ and $a=\gamma(\many{x})$
\item[--] $\delta(q_f, s) = \tru\,\,$ for any $s\in \alphabet'_{\FOLFPk}$
\end{itemize}

\paragraph{Correctness.} With regard to verifying true membership in a
defined relation, we can show by induction that the automaton always
has a run that reaches the final state $q_f$ along each branch. For
non-membership, on the other hand, the witness is that there is
\emph{no run witnessing the opposite}, i.e., that for every run there
is some branch that \emph{never reaches the final state}. It may not
be immediately obvious that such a witness can be finite
itself. However, any tuple must enter the computation within $n^r$
iterations. Thus it is sufficient to force the automaton to count down
from this bound while it checks that (for every run) some path does
not reach the final state. If the count reaches zero then there can be
no run witnessing membership.

It is more tedious to describe the language of the automaton at every
state (as was simple to do for the automaton without definitions and
least fixed points), given that it moves up and down on input
trees. The key for correctness is that, when the automaton is entering
a definition body from a down state
$\la\mathit{assign},j,\mathit{defn}\ra$, it checks that the membership
of the tuple of domain elements corresponding to $\mathit{assign}$ is
witnessed within the first $j$ stages of the least fixed point
computation for the definition. The dual case for checking
non-membership in the least fixed point is similar: the automaton
checks that membership is not witnessed in the first $j$ stages of the
computation. The finite height of $n^r$ of the powerset lattice for
definitions of arity $r$ and structure size $n$ gives the license to
stop when $j$ is $0$ and either allow the automaton to accept (if
checking non-membership) or reject (if checking membership).

\section{Term Synthesis Details}
\applabel{term-synthesis-appendix}

Here we give more details from ~\csecref{term-synthesis}: a ranked
alphabet for the language $\FOTERM$, the semantics for $\FOTERM$, and
the automaton transitions for reading terms in $\FOTERM$.

\subsection{Ranked Alphabet for $\FOTERMk$}
\applabel{foterm-alphabet}

A ranked alphabet for $\FOTERMk$ with definable functions from
$F=\{g_1,\ldots,g_{k_2}\}$:
\begin{align*}
  \alphabet_{\FOTERMk} &= \left\{
                         \Let\, g(\many{x})^2,\, \bigm\vert\,
                         g\in F,\,
                         \many{x}\in V^{\arity{g}}\right\} \cup \left\{ g^{\arity{g}}\right\} \cup \alphabet_{\FOLFP(k,k'-k_2)}
\end{align*}

\subsection{Semantics for $\FOTERM$}
\applabel{semantics-foterm}

The semantics for recursively-defined functions is as follows. Each
defined function $g$ of arity $d$ is interpreted as a partial function
$g^A : \dom(A)^d\rightharpoonup \dom(A)$, which is a member of the
bottomed partial order
$\mathcal{O} =\la \dom(A)^d\rightharpoonup \dom(a), \sqsubseteq,
\bot\ra$, where $\bot$ is undefined everywhere and $f\sqsubseteq f'$
holds if for all $\many{a}\in\dom(A)^d$, whenever
$f(\many{a})\downarrow$, then $f'(\many{a})\downarrow$ and
$f(\many{a})=f'(\many{a})$. This partial order has finite height since
we work with finite structures.  Now, suppose a function $g$ of arity
$d$ is defined recursively using a term $t(x_1,\ldots,x_d,g)$. We
associate a monotone function
$F_t : \mathcal{O}\rightarrow\mathcal{O}$ to the defining term $t$,
and let $g^A$ be the least fixed point of $F_t$, which can easily be
shown to exist and to be equal to the stable point of the finite chain
$\bot\sqsubseteq F_t(\bot)\sqsubseteq\cdots\sqsubseteq F^i_t(\bot) =
F^{i+1}_t(\bot)$, with $i\le |\dom(A)^d|$. The syntax and semantics
for terms in $\FOTERM$ guarantees monotonicity of the function $F_t$
for any term $t$, though we do not prove it here. It follows that
least fixed points exist for each definition. However, in general,
care is needed to ensure that the least fixed point is total on
$\dom(A)$, and whether or not this is so depends on the definition and
the structure $A$. Terms and formulas for $\FOTERM$ are interpreted in
a $3$-valued logic with interpretation functions
$\llbracket \varphi \rrbracket_{A,\gamma,D}$ and
$\llbracket t \rrbracket_{A,\gamma,D}$, where $A$ is a finite
structure, $\gamma$ is a partial variable assignment, and $D$ maps
definable symbols to the terms or formulas that define them. The
semantics is given in~\figref{forec-semantics}.

Partial functions disrupt the usual semantics of first-order logic
because we want to define entailment even when there are undefined
terms. As mentioned, we handle this with a 3-valued logic over
$\{\tru, \fals, \bot\}$, where $\bot$ means
\emph{undefined}. Undefinedness propagates across the various
formation rules of the logic as one might expect. If the classical
truth value of a formula cannot be determined based on the values of
its subformulas, then it has an undefined value. For example, the
formula $\varphi_0\wedge\varphi_1$ is undefined in a structure if
$\varphi_i$ has value $1$ and $\varphi_{1-i}$ is undefined
($i\in\{0,1\}$). (See ~\figref{forec-semantics} for the semantics.)

\paragraph{\bf Aside:} We note that another possibility for the
semantics, which would allow us to avoid a 3-valued logic, is to
interpret recursive functions as total functions on a lattice, which
could work as follows. Given an input structure $A$, the automaton
works over a modified structure $\Lat(A)$, which equips $A$ with a
lattice structure by introducing new elements $\bot,\top$ and putting
$x\le_{\Lat(A)} \top$ and $\bot\le x$ for all
$x\in\dom(A)$. Recursively-defined functions are then interpreted over
a partial order of functions
$F = \la \Lat(A)\rightarrow\Lat(A), \sqsubseteq\ra$, with
$f_1\sqsubseteq f_2$ if $f_1(x)\le_{\Lat(A)} f_2(x)$ for all
$x\in\Lat(A)$ and $f_1,f_2\in F$. If a definable function symbol $g$
is defined using term $t$, we can associate a function
$\mathcal{F}_t : F \rightarrow F$ and interpret $g$ as the least fixed
point of $\mathcal{F}_t$. The semantics of formulas is modified so
that quantification is defined only over elements in
$\dom(A)$. Furthermore, partial functions (e.g., $\mathit{car}$ from
the $\mathit{merge}$ example) can be modeled in $\Lat(A)$ as total
functions that use $\bot$ in the obvious way to model undefinedness,
and each function $f$ can be extended to have $f(\top) = \top$.

\begin{figure}
  \centering
  \scalebox{0.85}{
  \begin{minipage}{0.5\textwidth}
    \begin{align*}
      \llbracket R(\many{t}) \rrbracket_{A,\gamma,D} =
      \begin{cases}
        1 & \llbracket \many{t} \rrbracket_{A,\gamma,D}\downarrow,\,\,
        \llbracket \many{t} \rrbracket_{A,\gamma,D}\in R^A \\
        0 & \llbracket \many{t} \rrbracket_{A,\gamma,D}\downarrow,\,\,
        \llbracket \many{t} \rrbracket_{A,\gamma,D}\notin R^A \\
        \bot & \text{otherwise}
      \end{cases}
    \end{align*}
  \end{minipage}%
  \begin{minipage}{0.5\textwidth}
    \begin{align*}
      \llbracket \neg \varphi \rrbracket_{A,\gamma,D} =
      \begin{cases}
        1 & \llbracket \varphi \rrbracket_{A,\gamma,D} = 0 \\
        0 & \llbracket \varphi \rrbracket_{A,\gamma,D} = 1 \\
        \bot & \text{otherwise}
      \end{cases}
    \end{align*}
  \end{minipage}}

\scalebox{0.85}{
  \begin{minipage}{0.5\textwidth}
    \begin{align*}
      \llbracket \varphi_1\wedge\varphi_2 \rrbracket_{A,\gamma,D} =
      \begin{cases}
        1 & \llbracket \varphi_1 \rrbracket_{A,\gamma,D} = 1 \text{ and }
        \llbracket \varphi_2 \rrbracket_{A,\gamma,D} = 1 \\
        0 & \llbracket \varphi_1 \rrbracket_{A,\gamma,D} = 0 \text{ or }
        \llbracket \varphi_2 \rrbracket_{A,\gamma,D} = 0 \\
        \bot & \text{otherwise}
      \end{cases}
    \end{align*}
  \end{minipage}\hspace{0.2in}
  \begin{minipage}{0.5\textwidth}
    \begin{align*}
      \llbracket \varphi_1\vee\varphi_2 \rrbracket_{A,\gamma,D} =
      \begin{cases}
        1 & \llbracket \varphi_1 \rrbracket_{A,\gamma,D} = 1 \text{ or }
        \llbracket \varphi_2 \rrbracket_{A,\gamma,D} = 1 \\
        0 & \llbracket \varphi_1 \rrbracket_{A,\gamma,D} = 0 \text{ and }
        \llbracket \varphi_2 \rrbracket_{A,\gamma,D} = 0 \\
        \bot & \text{otherwise}
      \end{cases}
    \end{align*}
  \end{minipage}}

  \scalebox{0.85}{
  \begin{minipage}{0.5\textwidth}
    \begin{align*}
      \llbracket \forall x.\varphi \rrbracket_{A,\gamma,D} =
      \begin{cases}
        1 & \llbracket \varphi \rrbracket_{A,\update{\gamma}{x}{a},D} =
        1 \text{ all } a\in A \\
        0 & \llbracket \varphi \rrbracket_{A,\update{\gamma}{x}{a},D} =
        0 \text{ some } a\in A\\
        \bot & \text{ otherwise}
      \end{cases}
    \end{align*}
  \end{minipage}\hspace{0.2in}
  \begin{minipage}{0.5\textwidth}
    \begin{align*}
      \llbracket \exists x.\varphi \rrbracket_{A,\gamma,D} =
      \begin{cases}
        1 & \llbracket \varphi \rrbracket_{A,\update{\gamma}{x}{a},D} =
        1 \text{ some } a\in A \\
        0 & \llbracket \varphi \rrbracket_{A,\update{\gamma}{x}{a},D} =
        0 \text{ all } a\in A \\
        \bot & \text{ otherwise }
      \end{cases}
    \end{align*}
  \end{minipage}}

\scalebox{0.85}{
  \begin{minipage}{0.5\textwidth}
    \begin{align*}
      \llbracket \Let\,\,P(\many{x})=_{\lfp} \psi\,\, \In\,\, \varphi \rrbracket_{A,\gamma,D} &=
                                                                                                \llbracket \varphi \rrbracket_{A,\gamma,\update{D}{P\,}{\la \psi,D\ra}}
    \end{align*}
  \end{minipage}\hspace{0.2in}
  \begin{minipage}{0.5\textwidth}
    \begin{align*}
      \llbracket P(\many{t}) \rrbracket_{A,\gamma,D} =
      \begin{cases}
        1 & \llbracket \many{t} \rrbracket_{A,\gamma,D} \downarrow
        \text{ and } \llbracket \many{t} \rrbracket_{A,\gamma,D}\in \lfp(D(P)) \\
        0 & \llbracket \many{t} \rrbracket_{A,\gamma,D} \downarrow
        \text{ and } \llbracket \many{t} \rrbracket_{A,\gamma,D}\notin \lfp(D(P)) \\
        \bot & \text{ otherwise }
      \end{cases}
    \end{align*}
  \end{minipage}}

\vspace{0.2in}
\textsf{Formulas}\hspace*{0.85\textwidth}
\vspace{-0.1in}

\rule{0.95\textwidth}{0.4pt}

\textsf{Terms}\hspace*{0.885\textwidth}

\scalebox{0.85}{
  \begin{minipage}{0.2\textwidth}
    \begin{align*}
      \llbracket x \rrbracket_{A,\gamma,D} = \gamma(x)
    \end{align*}
  \end{minipage}
    \begin{minipage}{0.2\textwidth}
    \begin{align*}
      \llbracket c \rrbracket_{A,\gamma,D} = c^A
    \end{align*}
  \end{minipage}
    \begin{minipage}{0.4\textwidth}
    \begin{align*}
      \llbracket f(\many{t}) \rrbracket_{A,\gamma,D} =
      \begin{cases}
        f^A(\llbracket \many{t} \rrbracket_{A,\gamma,D}) & \llbracket
        \many{t} \rrbracket_{A,\gamma,D}\downarrow \\
        \bot & \text{ otherwise }
      \end{cases}
    \end{align*}
  \end{minipage}}
\scalebox{0.85}{
  \begin{minipage}{0.9\textwidth}
    \begin{align*}
      \llbracket \IfThenElse{\varphi}{t_1}{t_2} \rrbracket_{A,\gamma,D} =
      \begin{cases}
        \llbracket t_1 \rrbracket_{A,\gamma,D} & \llbracket
        \varphi \rrbracket_{A,\gamma,D} = 1 \\
        \llbracket t_2 \rrbracket_{A,\gamma,D} & \llbracket
        \varphi \rrbracket_{A,\gamma,D} = 0 \\
        \bot & \text{ otherwise }
      \end{cases}
    \end{align*}
  \end{minipage}}
\scalebox{0.85}{
    \begin{minipage}{0.5\textwidth}
    \begin{align*}
      \llbracket \Let\,\,g(\many{x})=_{\lfp} t\,\, \In\,\, t' \rrbracket_{A,\gamma,D} &=
        \llbracket t' \rrbracket_{A,\gamma,\update{D}{g\,}{\la t,D\ra}}
    \end{align*}
  \end{minipage}\hspace{0.1in}
  \begin{minipage}{0.5\textwidth}
    \begin{align*}
      \llbracket g(\many{t}) \rrbracket_{A,\gamma,D} =
      \begin{cases}
        \lfp(D(g))(\llbracket \many{t} \rrbracket_{A,\gamma,D}) &
        \llbracket \many{t} \rrbracket_{A,\gamma,D} \downarrow \\
        \bot & \text{ otherwise }
      \end{cases}
    \end{align*}
  \end{minipage}
}
\caption{Semantics for $\FOTERM$. $A$ is a finite structure, $\gamma$
  is a variable assignment, $D$ is a map from definable symbols to
  their defining terms or formulas, as well as the environment of
  definitions at the point where they are defined (similar to a
  closure). We write $\llbracket \many{t} \rrbracket$ to denote the
  tuple of interpretations of the terms in $\many{t}$, and
  $\llbracket t \rrbracket\downarrow$ to denote that an interpretation
  of a term is defined. }
\figlabel{forec-semantics}
\end{figure}

\subsection{Automaton for Evaluating Terms with Recursive Functions}
\applabel{automata-for-recursive-terms}

Fix a structure $A$ and a domain element $a\in\dom(A)$. We want to
define a two-way automaton
$\aut_A = \la Q, \alphabet_{\FOTERMk},I,\delta,F\ra$, that accepts the
set of terms $t\in\FOTERMk$ such that
$\llbracket t \rrbracket_{A,\emptyset,\emptyset} = a$. For simplicity
we assume $r$-ary functions and relations only. Similar to the
construction for $\FOLFP$. The state space, denoted
$Q_{\mathsf{TERM}}$, includes the states for $\FOLFP$ except now it
also includes states that encode the element which the current term
should evaluate to, as follows:
\begin{align*}
  Q_{\mathsf{TERM}}&\coloneq
                     \mathsf{EvalForm}\cup\mathsf{EvalTerm}\cup\{q_f\} \\
  \mathsf{EvalForm}&\coloneq (\mathit{Dual}(\mathit{Assign})\cup \mathit{Dual}(\mathit{Val}))\times \mathit{Count}\times
                        \mathit{Defn} \\
  \mathsf{EvalTerm}&\coloneq (\mathit{Assign} \cup \mathit{Val})\times\mathit{Count}\times\mathit{Defn}\times\dom(A)
\end{align*}
where $\mathit{Assign} = V\rightharpoonup\dom(A)$ is the set of
partial variable assignments, $\mathit{Val} = \dom(A)^r$ is a set of
values that are passed up to definitions of functions and relations,
$\mathit{Dual}(X) \triangleq \{x,\tilde{x} \,\mid\, x\in X\}$ helps us
define sets of symbols equipped with ``dual'' marked copies of their
members, and $\mathit{Count} = \{0,\ldots,n^r\}$ and
$\mathit{Defn} = \{\bot,P_1,\ldots,g_{k_2}\}$ serve the same purposes
as before. Below, we give the main transitions that allow the
automaton to evaluate terms. We use $\star$ as a wildcard to range
over any element from a set, determined by context. Here we denote
$\{1,\ldots,r\}$ with $[r]$.

\paragraph{\textbf{Reading variables.}} The automaton simply checks
that the variable is mapped to $a$ in the current assignment:
\vspace{-0.2in}
\begin{align*}
  \delta(\la\gamma,\star_1,\star_2,a\ra, x)  \,\,&=\,\,
                                                          \begin{cases}
                                                            \tru &
                                                            \gamma(x)
                                                            = a \\
                                                            \fals &
                                                            \text{
                                                              otherwise
                                                            }
                                                            \end{cases}
\end{align*}

\paragraph{\textbf{Reading \ite~ terms.}} The automaton must either (i) verify
the condition formula in the first child \emph{and} verify that the
term in the ``then'' branch (second child) evaluates to $a$ or (ii)
falsify the condition formula \emph{and} verify that the term in the
``else'' branch (third child) evaluates to $a$:
\begin{align*}
  \delta(\la\gamma,\star_1,\star_2,a\ra, \iteterm)  \,\,&=\,\, (\la\gamma,\star_1,\star_2\ra,1)\wedge
                                          (\la\gamma,\star_1,\star_2,a\ra,2) \vee
                                          (\la\tilde{\gamma},\star_1,\star_2\ra,1)\wedge
                                          (\la\gamma,\star_1,\star_2,a\ra,3)\end{align*}

\paragraph{\textbf{Reading function symbols.}} The automaton must verify that
the $r$ argument terms evaluate to some tuple
$\la a_1,\ldots,a_r\ra \in \dom(A)^r$ such that
$f^A(a_1,\ldots,a_r) = a$:
\begin{align*}
    \delta(\la\gamma,\star_1,\star_2,a\ra, f) \quad &=
    \bigvee_{\substack{\many{a}\,\in\, \dom(A)^r \\
    \text{\emph{s.t.}}\,\,f^A(\many{a}) = a}} \left(\bigwedge_{i\in
    [r]}(\la\gamma,\star_1,\star_2,a_i\ra,i)\right) \\
\end{align*}

\paragraph{\textbf{Reading relation symbols.}} The automaton must verify that the
$r$ argument terms evaluate to some tuple
$\la a_1,\ldots,a_r\ra \in R^A \subseteq \dom(A)^r$:
\begin{align*}
    \delta(\la\gamma,\star_1,\star_2,a\ra, R) \quad &=
    \bigvee_{\substack{\many{a}\,\in\, R^A}} \left(\bigwedge_{i\in
    [r]}(\la\gamma,\star_1,\star_2,a_i\ra,i)\right) \\
\end{align*}

\paragraph{\textbf{Reading a defined function $g$.}} The automaton
``guesses and verifies'' that the argument terms for $g$ evaluate to
$\many{a}$ and ascends to the definition of $g$ to verify that it
evaluates to $a$ when applied to $\many{a}$. Depending on the current
definition it either decrements or resets the counter.
\begin{align*}
  \delta(\la\gamma,j,g',a\ra, g) \quad &=
                                    \bigvee_{\substack{\many{a}\,\in\, \dom(A)^r}}
                                    \left(\bigwedge_{i\in
                                    [r]}(\la\gamma,j,g',a_i\ra,i)\right)
                                    \wedge (\la\many{a},n^r,g,a\ra,-1)
                                    \quad\text{ for }\, g'\neq g \\
  \delta(\la\gamma,j,g,a\ra, g) \quad &=
                                    \bigvee_{\substack{\many{a}\,\in\, \dom(A)^r}}
                                    \left(\bigwedge_{i\in
                                    [r]}(\la\gamma,j,g,a_i\ra,i)\right)
                                        \wedge (\la\many{a},j-1,g,a\ra,-1)
                                        \quad\text{ for }\, j > 0 \\
  \delta(\la\gamma,0,g,a\ra, g) \quad &=\quad \fals
\end{align*}

\paragraph{\textbf{Reading a defined relation $P$.}} The automaton
``guesses and verifies'' that the argument terms for $P$ evaluate to
$\many{a}$ and ascends to the definition of $P$ to verify that
$\many{a}$ is a member of the defined relation. Depending on the
current definition it either decrements or resets the counter.
\begin{align*}
    \delta(\la\gamma,j,P_l\ra, P_i) \quad &=
    \bigvee_{\substack{\many{a}\,\in\, \dom(A)^r}} \left(\bigwedge_{i\in
    [r]}(\la\gamma,j,P_l,a_i\ra,i)\right) \wedge
                                            (\la\many{a},n^r,P_i\ra,-1)
                                            \quad\text{ for }\, P_i\neq P_l \\
    \delta(\la\gamma,j,P_i\ra, P_i) \quad &=
    \bigvee_{\substack{\many{a}\,\in\, \dom(A)^r}} \left(\bigwedge_{i\in
    [r]}(\la\gamma,j,P_i,a_i\ra,i)\right) \wedge
                                            (\la\many{a},j-1,P_i\ra,-1)
                                            \quad\text{ for }\, j > 0
  \\
  \delta(\la\gamma,0,P_i\ra, P_i) \quad &=\quad \fals
\end{align*}

The rest of the transitions are similar in spirit to those for
$\FOLFP$. Initial states: $I=\{\la \emptyset,0,\bot,a\ra\}$. The
acceptance condition is again reachability, with $F=\{q_f\}$, and an
analysis of the automaton size is similar.

\section{Details from lower bounds ($\csecref{lower})$}
\applabel{reduction-details}

We describe the important parts of the grammar $G$ for the proof
of~\thmref{fokhardness}. A given sentence in the grammar can be viewed
as a tree of constraints. Branches of the tree assert various things
in order to reflect computation trees for the machine $M$. We have to
make extensive \emph{reuse} of variables in order to accumulate
polynomially-many constraints using a fixed number of variables. One
can picture a lopsided ``tree of constraints'' growing to one side
with nested conjunctions that requantify old variables as needed in
order to assert new constraints.

We use upper case words as names of nonterminals, sometimes abusing
notation by indicating the free variables common to any formula in the
language of the nonterminal, e.g., for a nonterminal $X$ the
expression $X(y)$ simply indicates that nonterminal $X$ is being
referred to, and any formula it generates has free variable $y$.

The grammar begins with start symbol $S$, which prepares some
variables needed later:
\begin{align*}
  \production{S} \produces \exists\,c\,c'\,\mathit{turn}\,.\,c=c'\wedge
  (\exists\, y. \mathit{Start}(y)\wedge \conj{\mathit{Init_0}(y,c,c',\mathit{turn})}{\mathit{Game}(y,c,c',\mathit{turn})})
\end{align*}

The variables $c,c'$ actually stand for triples of variables
$c_1,c_2,c_3$ and $c'_1,c'_2,c'_3$. These will be used to store
three-cell ``windows'' of the machine's tape. $\mathit{Start}(y)$
holds only for the unique special element, denoted $\star$ in
$\csecref{lower}$, which represents the beginning of the cycle in each
structure. Thus this causes $y$ to ``point'' to the beginning of the
cycle in each structure.

$\mathit{Init_0}$ is a nonterminal that effects the initialization of
the starting configuration by iteratively accumulating equalities
between variables and tape symbols corresponding to the starting
configuration. We skip this, since it is similar to later nonterminals
for generating configurations.

$\mathit{Game}(y,c,c',\mathit{turn})$ is the starting point for the game
semantics of alternating Turing machines. Two players, Adam and Eve,
take turns picking from one of two possible machine transitions. Eve
must be able to pick transitions in such a way that for any strategy
of Adam, the machine eventually halts in an accepting state. The
$\mathit{Game}$ nonterminal is as follows:
\begin{align*}
  \production{\mathit{Game}} \produces \mathit{Accepting}(c') \,\mid\,\mathit{MakeMove}(y,c,c',\mathit{turn})
\end{align*}
This is the first choice available to a synthesizer. It either picks a
formula from $\mathit{Accepting}$, in which case the machine should be
in an accepting state, or it picks a formula from $\mathit{MakeMove}$,
which accumulates constraints that simulate the game semantics.

The $\mathit{Accepting}$ nonterminal is simple:
\begin{align*}
    \production{\mathit{Accepting}} \produces \mathit{EncodesState}(c')\rightarrow \mathit{AcceptState}(c')
\end{align*}
It asserts that if the window of tape symbols represented in $c'$
contains a symbol that represents a machine state, then that state is
an accepting state (We don't spell this out, but it is simple to do
so.)

The other option is $\mathit{MakeMove}$:
\begin{align*}
    \production{\mathit{MakeMove}} \produces
  &\quad\iif{\mathit{turn}=\mathit{eve}} \\
  &\quad\then{\mathit{EveMove}(y,c,c')} \\
  &\quad\eelse{\mathit{AdamMove}(y,c,c')}
\end{align*}
The $\mathit{MakeMove}$ nonterminal branches on whether the current
move (which is kept in a variable) is for Adam or for Eve and
simulates the appropriate player's move. Note, here we use
$\dblqt{\iif \,\varphi\, \then{\varphi_1}\,\eelse{\varphi_2}}$ as a
shorthand for
$\varphi\rightarrow
\varphi_1\,\wedge\,\neg\varphi\rightarrow\varphi_2$ (and similarly
variants like
$\dblqt{\ifthen {\varphi}{\varphi'}\elifthen{\varphi_1}{\varphi_2}}$).

For an $\mathit{EveMove}$:
\begin{align*}
  \production{\mathit{EveMove}} \produces &\quad
                                         \,\exists\,\mathit{move}. \mathit{Move}(\mathit{move}) \\
  &\quad\wedge \exists\,\mathit{turn}. \mathit{turn}=\mathit{adam} \\
    &\quad\wedge \mathit{GenerateConfig}(y,c,c',\mathit{move},\mathit{turn})
\end{align*}
where $\mathit{Move}$ is another chance for the synthesizer to make a
choice, namely, which of two transitions $0$ or $1$ Eve should take:
\begin{align*}
    \production{\mathit{Move}} \produces \mathit{move} = 0\,\mid\,\mathit{move} = 1
\end{align*}
The next part of $\mathit{EveMove}$ requantifies $\mathit{turn}$ to
give Adam the next turn and then enters a polynomial-size gadget
called $\mathit{GenerateConfig}$ whose purpose is to let the
synthesizer choose tape symbols for the next machine configuration and
to accumulate constraints that force each structure to equate their
``window variables'' with the symbols for their window.
\begin{align*}
    \production{\mathit{GenerateConfig}} \produces \exists\,c'. \mathit{Cell_0}
\end{align*}

We omitted $\mathit{Init_0}$ because it is similar to
$\mathit{Cell_0}$, which begins the process of producing the tape
symbols for the next machine configuration. Rather than bogging down
in corner cases, we instead give the nonterminal for
$\mathit{Cell_i}$, with $0 < i < s$:
\begin{align*}
    \production{\mathit{Cell_{i}}} \produces &\quad \exists \,x.\, \mathit{Choose}(x)\,
           \wedge\\ &\quad\ifthen{\mathit{pred}(\star,y)}{c'_3=x}\\
            &\quad\elifthen{y=\star}{c'_2=x}\\
            &\quad\elifthen{\mathit{pred}(y,\star)}{c'_1=x}\\
            &\quad\wedge\exists\,y'.\mathit{next}(y,y')\wedge\,\mathit{Cell_{i+1}}(y',c,c',\mathit{move},\mathit{turn})
\end{align*}
where $\mathit{pred}$ is an immediate predecessor relation on the
cycle of a structure and $\mathit{next}$ is the immediate successor
relation. We use $\star$ here as a constant for the special element on
the cycle (dark element in ~\figref{reduction-figure}) that marks the
center of the window that a structure should track. The nonterminal
$\mathit{Choose}$ is another choice for the synthesizer:
\begin{align*}
  \production{\mathit{Choose}} \produces (x=a_1)\,\mid\,\ldots\,\mid\,(x=a_t)
\end{align*}
wherein it picks which of the $a_1,\ldots,a_t$ tape symbols comes next
at position $i$ of the next configuration. This gadget is replicated
again for the next cell, i.e. $\mathit{Cell_{i+1}}$, and so on. The
nonterminal for the final cell, i.e., $\mathit{Cell_{s}}$, finishes
with a conjunct for a formula from a nonterminal
$\mathit{VerifyWindows}$:
\begin{align*}
    \production{\mathit{VerifyWindow}} \produces &\quad
                                                   \ifthen{\mathit{move}=0}{\delta(0,c_1,c_2,c_3,c'_2)} \\
                   &\quad\elifthen{\mathit{move}=1}{\delta(1,c_1,c_2,c_3,c'_2)} \\
                   &\quad\wedge\,\exists\,c. c=c'\,\wedge\,\mathit{Game}(y,c,c',\mathit{turn})
\end{align*}
which asserts that the (move dependent) transition relation holds on
the variables which hold the tape symbols for a given window. Finally,
the previous configuration window variables are requantified and made
equal to the current configuration window variables, and the game
continues.

Finally, we return to the nonterminal for $\mathit{AdamMove}$:
\begin{align*}
  \production{\mathit{AdamMove}} \produces
  \,\forall\,&\mathit{move}. (\mathit{move} = 0 \vee \mathit{move} =
               1) \\
  \rightarrow & (\exists\,\mathit{turn}. \mathit{turn}=\mathit{eve}\,
                \wedge \\
             &\iif{\mathit{move}=0} \\
             &\then
               \mathit{GenerateConfig}(y,c,c',\mathit{move},\mathit{turn}) \\
             &\eelse \mathit{GenerateConfig}(y,c,c',\mathit{move},\mathit{turn}))
\end{align*}
The fact that the grammar contains two synthesis obligations, one for
when Adam picks transition $0$ and another for transition $1$, is
crucial. It allows the synthesizer to build a sentence which witnesses
a computation tree that depends on the history of Adam transitions.

Note that, for clarity, the description of the grammar above used more
variables than strictly necessary. For instance, by encoding triples
of tape symbols as single domain elements, and by tracking the current
turn in the grammar rather than a variable, we can reduce the number
of variables needed to $5$.

\emph{\bf \emph{Correctness.}}  If $M$ has an accepting computation
tree on $w$ the synthesizer can follow the strategy described by the
computation tree to produce a sentence that is true in every
structure. This involves synthesizing each tape symbol of each
successive $M$ configuration along each branch of the computation tree
and picking the correct Eve transition for each Adam transition. If
there is a sentence $\varphi\in L(G)$ that is true in each structure,
then an accepting computation tree can be reconstructed by traversing
the sentence. Each window centered at cell $i$ of each configuration
constructed in this way must evolve in accordance with the transition
relation since the sentence is true structure $A_i$.

\section{Mutual Recursion in $\FOLFP$}
\applabel{folfpm}

In~\figref{folfpm} we give syntax for an extension to $\FOLFP$ that
allows for blocks of mutually recursive relation definitions. As
discussed in~\csecref{datalog}, an automaton for evaluating such
blocks of definitions needs state that increases with the number of
distinct relations in a block. This number is bounded by the number of
available definable symbols, given that reusing the same symbol for
two distinct definitions in a single block is ambiguous (and thus we
rule it out). The total amount of state an automaton needs to evaluate
blocks of definitions thereby becomes independent of the size of
formulas.

\begin{figure}[H]
  \centering
  \begin{align*}
    \Def &\Coloneqq \Let\,\, \Block \,\,\In\,\,
           \Def \,\,\mid\,\, \varphi \\
    \Block &\Coloneqq  P(\many{x}) =_{\lfp} \varphi \,\,\mid\,\,  P(\many{x})
             =_{\lfp} \varphi \,\,\, \MyAnd \,\,\, \Block \\
    \varphi &\Coloneqq R(\many{t}) \mid \varphi \vee \varphi \mid
              \varphi \wedge \varphi \mid \neg \varphi \mid \exists
              x.\varphi \mid \forall x.\varphi \mid P(\many{t}) \\
    t &\Coloneqq x \mid c \mid f(\many{t}) \mid
        \IfThenElse{\varphi}{t}{t'}
  \end{align*}
  \caption{Grammar for $\FOLFP$ extended with blocks of (mutually)
    recursive definitions, where $\Def$ is the starting nonterminal,
    $P$ ranges over a set of definable relation symbols, and $R$ and
    $f$ range over relations and functions from a signature $\tau$. }
  \figlabel{folfpm}
\end{figure}

\section{$\Datalog$ Synthesis}
\applabel{datalog-synthesis}

We consider a $\Datalog$ synthesis problem as follows. Fix a set of
variables $V=\{v_1,\ldots, v_{k}\}$, and let $\many{x},\many{y}$ range
over $V$. Let $\tau = \la R_1,\ldots,R_s \ra$ be a relational
signature. Each $R_i$ will be an extensional predicate, in $\Datalog$
parlance (encoding facts). Fix a set of \emph{intensional} predicate
symbols $P= \{P_1,\ldots,P_{k'}\}$ which the $\Datalog$ program should
define. The $\Datalog$ programs we consider have the form $(\Pi,Q)$,
where $\Pi$ is a set of rules and $Q\in P$ is a distinguished
intensional predicate. Logically, a \emph{rule} is a formula of the
form
$\rho \coloneq \exists
\many{x}. \bigwedge_i^l\psi_i(\many{y},\many{x})$, where each $\psi_i$
is either of the form $R_i(\cdot)$, $\neg R_i(\cdot)$, or
$P_j(\cdot)$. Rules $\rho$ of this kind are used to define the
intensional predicates $P_1,\ldots, P_{k'}$. The reader may recognize
a different syntax that makes this clear:
$P_j(\many{y}) \leftarrow \psi_1(\many{y},\many{x}),\ldots,
\psi_l(\many{y},\many{x})$. Note that rules may be recursive, that is,
$P_j$ may appear as one of the $\psi_i$'s.

The semantics of a $\Datalog$ program $(\Pi,Q)$ over a
$\tau$-structure $A$ is defined in terms of a \emph{simultaneous fixed
  point}. Suppose an intensional predicate $P_i$ has arity $m_i$ and
is defined by $t$ clauses
$P_i(\many{y})\leftarrow
\psi_1^j(\many{y},\many{x}),\ldots,\psi_{c_j}^j(\many{y},\many{x})$,
for $1\le j\le t$. Let
$X = 2^{A^{m_1}}\times \cdots\times 2^{A^{m_{k'}}}$. Then for a
particular interpretation of the intensional predicates
$U=(U_1,\ldots,U_{k'})\in X$ we can define an immediate consequence
operator $F_{P_i} : X\rightarrow 2^{A^{m_i}}$ for the predicate $P_i$:
\begin{align*}
  F_{P_i}(U) \,\,=\,\, \left\{\, \many{a}\in A^{m_i} \,\Bigm\vert\,
  (A,U) \models \bigvee_{j=1}^t\exists \many{x}_j. \psi_1^j(\many{a},\many{x}_j)\wedge\cdots\wedge\psi_{c_j}^j(\many{a},\many{x}_j) \,\right\}
\end{align*}
where $(A,U)\models\varphi$ denotes entailment when the intensional
predicates are interpreted as in $U$. Each $F_{P_i}$ is clearly
monotonic since no $P_i$ appears negatively in the rules of
$\Pi$. Thus we have a system of monotonic functions $\{F_{P_i}\}_i$
and the intensional predicates are interpreted according to the
simultaneous least fixed point. For more on simultaneous fixed points we
refer the reader to~\cite{Fritz2002}.

Similar to our handling of $\FOLFPk$, we consider a variant of
$\Datalog$ with only $k$ first-order variables and $k'$ intensional
predicate symbols, denoted $\Datalogk$. It is easy to define a
suitable ranked alphabet $\alphabet_{\Datalogk}$ (which is similar to
earlier alphabets). Consider the problem in~\probref{datalog}:
\begin{algorithm}
  \BlankLine
  \KwIn{$\left\la (A, \Pos, \Neg), P_i, G \right\ra$ where}

  \myinput{$A$ is a $\tau$-structure}

  \myinput{$P_i$ is a \emph{target} intensional predicate symbol}
  \myinput{$\Pos,\Neg$ are sets of tuples over
    $\dom(A)^{\arity{P_i}}$}

  \myinput{$G$ is an RTG over $\alphabet_{\Datalogk}$}

  \KwOut{$\Pi\in L(G) \,\, s.t. \,\, \Pos\subseteq \lfp(F_{P_i}) \text{
      and } \Neg\cap \lfp(F_{P_i}) = \emptyset$}

  \myoutput{Or ``No'' if no such program exists}
  \NoCaptionOfAlgo
  \caption{\textbf{Problem 4:} \sepsynth{$\Datalogk$}}
  \problabel{datalog}
\end{algorithm}

\begin{theorem}
  \sepsynth{$\Datalogk$}~is decidable in $\EXP$ for a fixed signature
  and fixed $k,k'\in\Nat$.
\end{theorem}
\begin{proof}[Sketch.]
  Follows very closely the construction from \secref{recursion}. The
  main novelty is to evaluate mutually recursive rules with a two-way
  automaton. Similar to the case of $\FOLFP$ we can use a counter in
  the automaton state to ensure that we correctly check non-membership
  in least fixed point relations. The difference is that we need to
  keep multiple counters at once (one per intensional predicate),
  because evaluating any single definition may involve evaluating
  other definitions that can refer back to the original. Note that our
  solution for $\FOLFP$ could get away with a single counter because
  each successive definition could only refer to previously defined
  relations. In contrast, there is no notion of scope within a single
  block, and all recursive definitions in the same block see each
  other.
\end{proof}

\else
\fi

\end{document}